\documentclass{article}

\usepackage[english]{babel}
\usepackage[utf8x]{inputenc}
\usepackage[T1]{fontenc}
\usepackage[a4paper, margin=1in]{geometry}
\usepackage{graphicx}
\usepackage[colorinlistoftodos]{todonotes}
\usepackage[colorlinks=true]{hyperref}
\usepackage{amsthm}
\usepackage{amsmath}
\usepackage{amssymb}
\usepackage{amsfonts}
\usepackage{mathrsfs}
\usepackage{mathtools}
\usepackage{verbatim}
\usepackage{footnote}
\usepackage[linesnumbered,boxed, vlined]{algorithm2e}
\usepackage{lineno}
\usepackage{caption}

\newtheorem{theorem}{Theorem}
\newtheorem{lemma}[theorem]{Lemma}
\newtheorem{corollary}[theorem]{Corollary}
\newtheorem{definition}[theorem]{Definition}

\newtheorem{obsv}[theorem]{Observation}

\newcommand{\calA}{\mathcal{A}}
\newcommand{\calB}{\mathcal{B}}
\newcommand{\calC}{\mathcal{C}}
\newcommand{\calD}{\mathcal{D}}
\newcommand{\calH}{\mathcal{H}}

\newcommand{\calO}{\mathcal{O}}

\newcommand{\calX}{\mathcal{X}}
\newcommand{\calY}{\mathcal{Y}}

\newcommand{\loss}[2]{L_{#1}({#2})}
\newcommand{\txtfrac}[2]{{#1}/{#2}}

\newcommand{\dist}{\operatorname{d}}
\newcommand{\distsq}{\operatorname{d^2}}
\newcommand{\euc}{\operatorname{EUC}}
\newcommand{\erm}{\operatorname{ERM}}
\newcommand{\fms}{\operatorname{FMS}}
\newcommand{\bin}{\operatorname{bin}}
\newcommand{\bb}{\operatorname{b}}

\newcommand{\avgdistsq}{\operatorname{\overline{\Delta}}}
\DeclareMathOperator*{\cost}{cost}
\DeclareMathOperator*{\error}{error}

\newcommand{\vcdim}{\mathop{\mathrm{VCdim}}\nolimits}   
\newcommand{\Ndim}{\mathop{\mathrm{Ndim}}\nolimits}

\newcommand{\poly}{\operatorname{poly}}
\newcommand{\E}{\mathop{\mathbb{E}}}

\DeclareMathOperator*{\argmin}{argmin}
\DeclareMathOperator*{\argmax}{argmax}

\newcommand{\R}{{\mathbb{R}}}

\newcommand\numberthis{\addtocounter{equation}{1}\tag{\theequation}}

\begin{document}

\title{Semi-Supervised Algorithms for 
     Approximately Optimal and Accurate Clustering}
\author{
  Buddhima Gamlath \\
  \texttt{buddhima.gamlath@epfl.ch}
    \and
  Sangxia Huang \\
  \texttt{huang.sangxia@gmail.com}
    \and
  Ola Svensson \\
  \texttt{ola.svensson@epfl.ch}
}\date{}

\maketitle


\begin{abstract}

We study $k$-means clustering in a semi-supervised setting. 
Given an oracle that returns whether two given points belong to the same
cluster in a fixed optimal clustering, we investigate the following question:
how many oracle queries are sufficient to efficiently recover a clustering 
that, with probability at least $(1 - \delta)$, simultaneously has a cost of 
at most $(1 + \epsilon)$ times the optimal cost and an accuracy  of at least 
$(1 - \epsilon)$?

We show how to achieve such a clustering on $n$ points with 
$O{((k^2 \log n) \cdot m{(Q, \epsilon^4, \delta / (k\log n))})}$ 
oracle queries, when the $k$ clusters can be learned with an $\epsilon'$ 
error and a failure probability $\delta'$ using $m(Q, \epsilon',\delta')$
labeled samples in the supervised setting, where $Q$ is the set of 
candidate cluster centers. 
We show that $m(Q, \epsilon', \delta')$ is small both for $k$-means 
instances in Euclidean space and for those in finite metric spaces.
We further show that, for the Euclidean $k$-means instances, we can avoid the
dependency on $n$ in the query complexity at the expense of an increased
dependency on $k$: specifically, we give a slightly more involved algorithm 
that uses 
$O{(k^4/(\epsilon^2 \delta) + (k^{9}/\epsilon^4) \log(1/\delta) +  
k \cdot m{(\R^r, \epsilon^4/k, \delta)})}$ 
oracle queries. 

We also show that the number of queries needed for 
$(1 - \epsilon)$-accuracy in Euclidean $k$-means must linearly depend 
on the dimension of the underlying Euclidean space, and for
finite metric space $k$-means, we show that it must at least be logarithmic 
in the number of candidate centers.
This shows that our query complexities capture the right dependencies 
on the respective parameters.

\end{abstract}


\section{Introduction}
\label{sec:intro}
Clustering is a fundamental problem that arises in many learning tasks. 
Given a set $P$ of data points, the goal  is to output a $k$-partition 
$C_1 \dot{\cup} \dots \dot{\cup} C_k$ of $P$ according to some optimization 
criteria.
In unsupervised clustering, the data points are unlabeled. 
The classic $k$-means problem and other well-studied clustering problems
such as $k$-median fall into this category.

In a general $k$-means clustering problem, the input comprises a finite set
of $n$ points $P$ that is to be clustered, a set of candidate centers $Q$, 
and a distance metric $\dist$ giving the distances between each pair of 
points in $P \cup Q$.
The goal is to find $k$ cluster centers $c_1, \dots, c_k \in Q$ that 
minimizes the cost, which is the sum of squared distances
between each point in $P$ and its closest cluster center. 
In this case, the clustering $\calC$ is defined by setting 
$C_i = \{x \in P: c_i \mbox{ is the closest center to } x \}$ 
for all $i = 1, \dots, k$ and breaking ties arbitrarily. 
Two widely studied special cases are the $k$-means problem in 
Euclidean space (where $P \subset \R^r, Q = \R^r$, and $\dist$ is 
the Euclidean distance function) and the $k$-means problem 
in finite metric spaces (where $(P \cup Q, \dist)$ forms a finite metric space).

Despite its popularity and success in many settings, there are two known 
drawbacks of the 
{unsupervised} $k$-means problem:
\begin{enumerate}
\item Finding the centers that satisfy the clustering goal is computationally 
hard. For example, even the special case of $2$-means problem in Euclidean space 
is NP-hard~\cite{Das08}. 
\item There could be multiple possible sets of centers that minimize the cost. 
However, in practical instances, not all such sets are equally meaningful, 
and we would like our algorithm to find one that corresponds to the concerns 
of the application.
\end{enumerate}

Since $k$-means is NP-hard, it is natural to seek approximation algorithms.
For the general $k$-means problem in Euclidean space, notable approximation
results include the local search by Kanungo et al.~\cite{KMNPSW02} with an 
approximation guarantee of $(9+\epsilon)$ and the recent LP-based  
$6.357$-approximation algorithm by Ahmadian et al.~\cite{ANSW17}. 
On the negative side, Lee et al.~\cite{LSW17} ruled out arbitrarily good 
approximation algorithms for the $k$-means problem on general instances. 
For several special cases, however, there exist PTASes. 
For example, in the case where $k$ is constant, Har-Peled and 
Mazumdar~\cite{HM04} and Feldman et al.~\cite{FMS07} showed how to get a 
PTAS using weak coresets, and in the case where the dimension $d$ is 
constant, Cohen-Addad et al.~\cite{CKM16} and 
Friggstad et al.~\cite{FRS16} gave PTASes based on a basic local 
search algorithm. In addition, Awasthi et al.~\cite{ABS10} 
presented a PTAS for $k$-means, assuming that the input is ``clusterable'' 
(satisfies a certain stability criterion).

Even if we leave aside the computational issues with unsupervised $k$-means, 
we still have the problem that there can be multiple different clusterings that 
minimize the cost. 
To see this, consider the $2$-means problem on the set of vertices of an 
equilateral triangle. 
In this case, we have three different clusterings that give the same minimum 
cost, but only one of the clusterings might be meaningful. 
One way to avoid this issue is to have strong 
assumptions on the input. 
For example, Balcan et al.~\cite{BBG13} considered the problem in a restricted 
setting where any $c$-approximation to the problem also 
classifies at least a $(1 - \epsilon)$ fraction of the points correctly.

Ashtiani et al.~\cite{AKB16} recently proposed a different approach for 
addressing the aforementioned drawbacks.
They introduced a semi-supervised active clustering framework where the 
algorithm is allowed to make queries of the form \emph{same-cluster(x, y)} 
to a domain expert, and the expert replies whether the points $x$ and $y$ 
belong to the same cluster in some fixed optimal clustering. 
Under the additional assumptions that the clusters are contained inside $k$ 
balls in $\mathbb{R}^r$ that are sufficiently 
far away from each other, they presented an algorithm that makes 
$O(k^2 \log k + k (\log n + \log (1/\delta)))$
same-cluster queries, runs in $O(kn \log n + k^2 \log (1/\delta))$ time, 
and recovers the clusters with probability at least $(1 - \delta)$. 
Their algorithm finds approximate cluster centers, orders 
all points by their distances to the cluster centers, and performs 
binary searches to determine the radii of the balls. 
Although it recovers the exact clusters, this approach works only 
when the clusters are contained inside well-separated balls. 
When the clusters are determined by a general Voronoi 
partitioning, and thus distances to the cluster boundaries can differ in 
different directions, this approach fails.

A natural question arising from the work of Ashtiani et al.~\cite{AKB16} is 
whether such strong assumptions on the input structure are necessary. 
Ailon et al.~\cite{ABJK18}  addressed this concern and considered the 
problem without any assumptions on the structure of the underlying true 
clusters.
Their main result was a polynomial-time $(1+\epsilon)$-approximation scheme 
for $k$-means in the same semi-supervised framework as in 
Ashtiani et al.~\cite{AKB16}. 
However, in contrast to Ashtiani et al.~\cite{AKB16}, their work gives no 
assurance on the  accuracy of the recovered clustering compared to the true 
clustering. 
To achieve their goal, the authors utilized importance sampling to 
uniformly sample points from small clusters that significantly contribute 
to the cost. 
Their algorithm makes $O(k^{9}/\epsilon^4)$ 
same-cluster queries,  runs in $O(nr(k^{9}/\epsilon^4))$ time, 
and succeeds with a constant probability.

In this work, we investigate the $k$-means problem in the same 
semi-supervised setting as Ailon et al.~\cite{ABJK18}, but in addition 
to approximating the cost, we seek a solution that is also accurate with 
respect to the true clustering. 
We assume that the underlying true clustering minimizes the cost, 
and that there are no points on cluster boundaries (i.e., the margin 
between each pair of clusters can be arbitrarily small but not zero). 
This last assumption is what differentiates our setup from that of 
Ailon et al.~\cite{ABJK18}. 
It is reasonable to assume that no point lies on the boundary of two 
clusters, as otherwise, to achieve constant accuracy, we would have to 
query at least a constant fraction of the boundary points. 
Without querying each boundary point, we have no way of determining to 
which cluster it belongs.

Observe that if we label all the points correctly with respect to the true 
clustering, the resulting clustering automatically achieves the optimal cost. 
However, such perfect accuracy is difficult to achieve as there may be points 
that are arbitrarily close to each other but belong to different clusters.
Using only a reasonable number of samples, the best we can hope 
for is to recover an approximately accurate solution.  
PAC (Probably Approximately Correct) learning helps us achieve this goal and 
provides a trade-off between the desired accuracy  and the required number of 
samples.

Suppose that we have a solution where only a small fraction of the input points 
is incorrectly classified. 
In this case, one would hope that the cost is also close to the optimal cost.
Unfortunately,
the extra cost incurred by the incorrectly classified points can be 
very high depending on their positions, true labels, and the labels 
assigned to them. 
Our main concern in this paper is controlling this additional cost.

We show that if we start with a constant-factor approximation for the cost, we 
can refine the clustering using a PAC learning algorithm. 
This yields a simple polynomial-time algorithm that, given a $k$-means instance 
and $(\epsilon, \delta) \in (0, 1)^2$ as parameters, with probability at least 
$(1 - \delta)$ outputs a clustering that has a cost of at most $(1 + \epsilon)$ 
times the optimal cost and that classifies  at  least a $(1 - \epsilon)$ 
fraction of the points correctly with respect to the underlying true clustering.  
To do so, the algorithm makes 
$O((k^2 \log n) \cdot m(Q, \epsilon^4, \delta/(k \log n)))$
same-cluster queries. Here, $m(Q, \epsilon', \delta')$ is the sufficient number 
of labeled samples for a PAC learning algorithm to learn $k$ clusters in a 
$k$-means instance with an $\epsilon'$ error and a failure probability 
$\delta'$ in the supervised setting 
(recall that $Q$ is the set of candidate centers). 
We further show that our algorithm  can be easily 
adapted to $k$-median and other similar problems that use the $\ell$'th power of 
distances in place of squared distances for some fixed $\ell > 0$.
We formally present this result as Theorem~\ref{thm:simplealg} in 
Section~\ref{sec:simplealg}.
In Theorem~\ref{thm:simplealg-informal} below, we give an informal statement 
for the case of $k$-means.

\begin{theorem}[An informal version of Theorem~\ref{thm:simplealg}]
\label{thm:simplealg-informal} 
There exists a semi-supervised learning algorithm that, given a 
$k$-means instance, oracle access to same-cluster queries that are 
consistent with some fixed optimal clustering, 
and parameters $(\epsilon, \delta) \in (0, 1)^2$,
outputs a clustering that, with probability at least $(1 - \delta)$, 
correctly labels (up to a permutation of the labels) at least a 
$(1 - \epsilon)$ fraction of the points and, simultaneously, has a cost 
of at most $(1 + \epsilon)$ times the optimal cost. 
In doing so, the algorithm makes 
$O((k^2 \log n) \cdot m(Q, \epsilon^4, \delta/(k \log n)))$ 
same-cluster queries.
\end{theorem}

Our algorithm is general and applicable to any family of $k$-means, 
$k$-median, or similar distance based clustering instances that can be 
efficiently learned with PAC learning. 
As shown in Appendix~\ref{app:pac}, these include Euclidean 
and general finite metric space clustering instances. 
In contrast, both Ashtiani et al.~\cite{AKB16} and Ailon et al.~\cite{ABJK18}, 
considered only the Euclidean $k$-means problem.
To the best of our knowledge, ours is the first such result applicable to 
finite metric space $k$-means and both Euclidean and finite metric space 
$k$-median problems.

Ideally, we want $m(Q, \epsilon, \delta)$ to be small.
Additionally, the analysis of our algorithm relies on two natural properties 
of learning algorithms.
Firstly, we require PAC learning to always correctly label all 
the sampled points. 
Secondly, we also require it to not `invent' new labels and only output 
labels that it has seen on the samples.
We show that such learning algorithms with small $m(Q, \epsilon, \delta)$ 
exist both for $k$-means instances in Euclidean space and for those in 
finite metric spaces with no points on the boundaries of the optimal clusters. 
For $r$-dimensional Euclidean $k$-means, $m(Q = \R^r, \epsilon, \delta)$ has 
a linear dependency on $r$. 
For the case of finite metric spaces, $m(Q, \epsilon, \delta)$ 
has a logarithmic dependency on $|Q|$, which is the size of the set of 
candidate centers.
In fact, these learning algorithms are applicable not only to 
$k$-means instances but also to instances of other similar center-based 
clustering problems (where clusters are defined by assigning points
to their closest cluster centers). 

Our semi-supervised learning algorithm is inspired by the work of Feldman 
et al.~\cite{FMS07} on weak coresets. 
Their construction of the weak coresets first obtains an intermediate 
clustering using a constant-factor approximation algorithm and refines 
each intermediate cluster by taking random samples. 
In order to get a good guarantee for the cost, their algorithm partitions 
each cluster into an inner ball that contains the majority of the points, 
and an outer region that contains the remaining points.
We proceed similarly to this construction; however, we further partition the 
outer region into $O(\log n)$ concentric rings and use PAC learning to label 
the points in the inner ball and in each of the outer rings separately. 
For Euclidean $k$-means instances, the number of same-cluster queries needed 
by the algorithm has a logarithmic dependency on the number $n$ of points, 
which is similar (up to a $\poly(\log \log n)$ factor) to that of the 
algorithm by Ashtiani et al.~\cite{AKB16}.  
The advantage of our algorithm is that it works for a much broader range of 
$k$-means instances whereas the applicability of the algorithm of Ashtiani 
et al.~\cite{AKB16} is restricted to those instances whose clusters 
are contained in well-separated balls in Euclidean space. 

This algorithm is effective in many natural scenarios where the number of 
clusters $k$ is larger than $\log n$. 
However, as the size of the $k$-means instance (i.e., the number of points) 
becomes large, the $\log n$ factor becomes undesirable. 
In Euclidean $k$-means, the number of samples needed by the learning 
algorithm for an $\epsilon$ error and a failure probability $\delta$ does not 
depend on $n$. 
The $\log n$ dependency in the final query complexity is exclusively due to
repeating the PAC learning step on $\Omega(k \log n)$ different partitions 
of $P$. 
To overcome this problem, we present a second algorithm, which is applicable 
only to Euclidean $k$-means instances, inspired by the work 
of Ailon et al.~\cite{ABJK18}. 
This time, we start with a $(1 + \epsilon)$-approximation for the cost and 
refine it using PAC learning. 
Unlike our first algorithm, we only run the PAC learning once on the whole 
input, and thus we completely eliminate the dependency on $n$. 
The disadvantages of this algorithm compared to our first algorithm are the 
slightly more involved nature of the algorithm and the increased dependency 
on $k$ in its query complexity.  Theorem~\ref{thm:indep-n} below formally
states this result. 
The proof follows from the analysis our algorithm in Section~\ref{sec:alg2}.

\begin{theorem}
\label{thm:indep-n}
There exists a polynomial-time algorithm that, given a $k$-means instance
in $r$-dimensional Euclidean space, oracle access to same-cluster queries  
that are consistent with some fixed optimal clustering, and parameters 
$(\epsilon, \delta) \in (0, 1)^2$,
outputs a clustering that, with probability at 	least $(1 - \delta)$, 
correctly labels 
(up to a permutation of the labels) at least a $(1 - \epsilon)$ fraction 
of the points and, 
simultaneously, has a cost of at most $(1 + \epsilon)$ times the optimal cost.
The algorithm makes $O( k^4/(\epsilon^2 \delta) + 
(k^{9}/\epsilon^4) \log (1/\delta ) +  
\allowbreak k \cdot m (\R^r, \epsilon^4/k, \delta) )$ 
same-cluster queries.
\end{theorem}

For the Euclidean setting, the query complexities of both our algorithms
have a linear dependency on the dimension of the Euclidean space. 
The algorithm of Ashtiani et al.~\cite{AKB16} does not have such a dependency 
due to their strong assumption on the cluster structure, whereas the one by 
Ailon et al.~\cite{ABJK18} does not have that as it only approximates the cost. 
We show that, in our scenario, such a dependency is necessary to achieve the 
accuracy guarantees of our algorithms.
For the finite metric space $k$-means, the query complexity of our general 
algorithm has an $O(\poly (\log |P|,  \log |Q|))$ dependency. 
The dependency on $|P|$ comes from the repeated application of the learning 
algorithm on $\Omega(k \log |P|)$ different partitions, and whether we can 
avoid this is an open problem. 
However, we show that an $\Omega(\log |Q|)$ query complexity is necessary for 
the accuracy. 
Formally, we prove the following theorem in Section~\ref{sec:lbs}.

\begin{theorem}
\label{thm:lbs}
Let $K$ be a family of $k$-means instances. Let $\calA$ be an algorithm that, 
given a $k$-means instance in $K$, oracle access to same-cluster queries
for some fixed optimal clustering, and parameters 
$(\epsilon, \delta) \in (0, 1)^2$,
outputs a clustering that, with probability at least $(1 - \delta)$, 
correctly labels (up to a permutation of the cluster labels) at least a 
$(1 - \epsilon)$ fraction of the points.
Then, the following statements hold:
\begin{enumerate}
\item If $K$ is the family of $k$-means instances in $r$-dimensional 
Euclidean space that have no points on the boundaries of optimal clusters, 
$\calA$ must make $\Omega(r)$ same-cluster queries.
\item If $K$ is the family of finite metric space $k$-means instances 
that have no points on the boundaries of optimal clusters, $\calA$ must make
$\Omega(\log |Q|)$ same-cluster queries.
\end{enumerate} 
\end{theorem}

The outline of this extended abstract is as follows. 
In Section~\ref{sec:prelim} we introduce the notation, formulate the problem 
and present the learning theorems that we use in the subsequent 
sections. 
In Section~\ref{sec:simplealg} we present our first algorithm, which is 
simple and applicable to general $k$-means instances that admit efficient 
learning algorithms, but has a dependency of $\log n$ in its query complexity. 
In Section~\ref{sec:alg2} we discuss how to remove the $\log n$ dependency 
in the query complexity for the special case of Euclidean $k$-means instances 
and present our second algorithm. 
In Section~\ref{sec:lbs}, we prove our query lower bound claims of Theorem~\ref{thm:lbs}. 
In Appendix~\ref{app:pac}, we introduce the basic concepts and tools of PAC 
learning and explain how to design learning algorithms for Euclidean and finite
metric space $k$-means instances.


\section{Preliminaries}
\label{sec:prelim}

In this section, we introduce the basic notation and two common 
families of $k$-means instances, and formally define the $k$-means problem 
that we address in this work.
We also introduce the notion of \emph{learnability} for families of $k$-means
instances and state two learning theorems that will be used in the 
later sections.

\subsection{k-Means Problem in a Semi-supervised Setting}
\label{sec:k-means}

Let $P$ and $Q$ be two sets of points where $|P| = n$, and let 
$\dist:(P \cup Q) \times (P \cup Q) \to \R_+$ be a distance metric.
We denote a $k$-means instance by the triple $(P, Q, \dist)$.
Two common families of $k$-means instances we consider in this work are: 
\begin{enumerate}
\item $k$-means instances in Euclidean space, where $P \subset \R^r$, 
$Q = \R^r$, and $\dist(x_1, x_2) = \|x_1 - x_2\|$ is the Euclidean 
distance between $x_1$ and $x_2$, and
\item $k$-means instances in finite metric spaces, where $(P \cup Q, \dist)$ 
forms a finite metric space.
\end{enumerate}

Let $[k] := \{1, \dots, k\}$. We identify a $k$-clustering $\calC$ of 
$(P, Q, \dist)$ by a labeling function $f_\calC:P \to [k]$, and a set of $k$ 
centers, $c_1, \dots, c_k \in Q$, associated with each label, $1, \dots, k$.
For each label $i \in [k]$ of a clustering $\calC$, let 
$C_i := \{p \in P : f_\calC(p) = i \}$ be the set of points whose label is $i$. 
For convenience, we may use the labeling function $f_\calC$ or the set 
of clusters $\{C_1, \dots, C_k\}$ interchangeably to denote a clustering 
$\calC$.

For a subset $C \subseteq P$ and a point $q \in Q$, define 
$\cost(C, q) := \sum_{p \in C} \distsq(p, q)$. 
For each $i$, define center 
$c_i := \argmin_{q \in Q} \cost(C_i, q)$, 
i.e., each center is a point in $Q$ that minimizes the sum of squared distances 
between itself and each of the points assigned to it.
For a $k$-clustering $\calC$, we define its \emph{k-means cost} as
$\cost(\calC) := \sum_{i \in [k]} \cost(C_i, c_i).$
Let $\calC^\ast$ be the set of all $k$-clusterings of $(P, Q, \dist)$.
Then, the optimal $k$-means cost of $(P, Q, \dist)$ is defined as
$ OPT  := \min_{\calC \in \calC^\ast} \cost(\calC).$
We say that a $k$-clustering $\calC$ $\alpha$-approximates the $k$-means cost if 
$\cost(\calC) \leq \alpha OPT$. 

Let $\calO$ be a fixed $k$-clustering of $(P, Q, \dist)$ that achieves the 
optimal $k$-means cost, and let $\calC$ be any $k$-clustering of $P$. 
Let $f_\calO$ and $f_\calC$ be the labeling functions that correspond to $\calO$ 
and $\calC$ respectively.
We assume that we have oracle access to the labeling function $f_\calO$ of the 
optimal target clustering up to a permutation of the labels.
We can simulate a single query to such an oracle with $O(k)$ queries to a 
same-cluster oracle as explained in Algorithm~\ref{alg:get-cluster}. 
A same-cluster oracle is an oracle that answers $same\text{-}cluster(p_1, p_2)$ 
queries with `yes' or `no' based on whether $p_1$ and $p_2$ belong to the 
same cluster in the fixed optimal clustering $\calO$.

The error of a clustering $\calC$ with respect to the clustering 
$\calO$ for a $k$-means instance $(P, Q, \dist)$ is now defined as
$ \error(\calC, \calO) := \min_\sigma | \{ p \in P : f_\calO(p) \neq 
\sigma(f_\calC(p))\} |$, where the
minimization is over all permutations $\sigma: [k] \to [k]$.
In other words, $\error(\calC, \calO)$ is the minimum number of points 
incorrectly labeled  by the clustering $\calC$ with respect to the 
optimal clustering $\calO$, considering all possible permutations of the 
cluster labels. 
The reason for defining $\error$ in this manner is because we use a 
simulated version of $f_\calO$ (which is only accurate up to a permutation 
of the cluster labels) instead of the true $f_\calO$ to learn 
cluster labels.
We say that a $k$-clustering $\calC$ is $(1 - \alpha)$-accurate with 
respect to $\calO$ if $\error(\calC, \calO) \leq \alpha n$.

\SetKwInOut{Global}{Global}

\begin{algorithm}
  \SetKwInOut{Input}{Input}\SetKwInOut{Output}{Output}
  \Input{A point $x \in X$, oracle access to $same\text{-}cluster(x_1, x_2)$.}
  \Output{A label $i \in [k]$.}
  \Global{A list of points $S = [ \, \, \,  ]$.}
  \BlankLine
  \For{$1 \leq i \leq length(S)$}{
    \If{$same\text{-}cluster(x, S[i])$}   {Return $i$}
  }
  Append $x$ to $S$.\\
  Return $length(S)$.\\
  \caption{Simulating a labeling oracle with the same-cluster oracle.}
  \label{alg:get-cluster}
\end{algorithm}

Given $(P, Q, \dist)$, parameters $k$ and 
$(\epsilon, \delta) \in (0, 1)^2$, and oracle access to $f_\calO$, our goal
is to output a $k$-clustering $\hat{\calO}$ of $(P, Q, \dist)$ that, with
probability at least $(1-\delta)$,  satisfies $\error(\hat{\calO}, \calO) 
\leq \epsilon n$ and $\cost(\hat{\calO}) \leq (1 + \epsilon) OPT$.

\subsection{PAC Learning for  k-Means}
\label{sec:pac-for-kmeans}

Let $K$ be a family of $k$-means instances, and let 
$m( Q, \epsilon, \delta)$ be a positive integer-valued function. 
We say such a family $K$ is \emph{learnable} with \emph{sample complexity} $m$ 
if there exists a learning algorithm $\calA_L$ such that the following holds: 
Let $\epsilon \in (0, 1)$ be an error parameter and let $\delta \in (0, 1)$ 
be a probability parameter. 
Let $(P, Q, \dist)$ be a $k$-means instance that belongs to $K$. 
Let $\calO$ be a fixed optimal $k$-means clustering and let $f_\calO$ be 
the associated labeling function. 
Let $T$ be a fixed subset of $P$, and let $S$ be a multiset of at least
$m( Q, \epsilon, \delta)$ independently and uniformly 
distributed samples from $T$. 
The algorithm $\calA_L$, given input $(P, Q, \dist)$ and $(s, f_\calO(s))$ 
for all $s \in S$, outputs a function $h: P \to [k]$. 
Moreover, with probability at least $(1 - \delta)$ over the choice of $S$, 
the output $h$ agrees with $f_\calO$ on at least a $(1 - \epsilon)$ fraction 
of the points in $T$ 
(i.e., $|\{ p \in T : h(p) = f_\calO(p)  \}| \geq (1 - \epsilon) |T|$).
This simpler notion of learnability is sufficient for the purpose of this 
work although it deviates from that of the general PAC learnability, which 
concerns with samples drawn from arbitrary distributions. 

We say that such a learning algorithm $\calA_L$ has the \emph{zero sample error} 
property if the output $h$ of $\calA_L$ assigns the correct label to all 
the sampled points (i.e., $h(s) = f_\calO(s)$ for all $s \in S$). 
Furthermore, we say that such a learning algorithm $\calA_L$ is 
\emph{non-inventive} if it does not `invent' labels that it has not seen. 
This means that the output $h$ of $\calA_L$ does not assign labels that were 
not present in the input (sample, label) pairs (i.e., if $h(x) = c$ for 
some $x \in P$, then for some sample point $s \in S$, $f_\calO(s) = c$).

In Section~\ref{sec:simplealg}, we present a simple algorithm for 
$(1 + \epsilon)$-approximate and $(1-\epsilon)$-accurate $k$-means clustering 
for a family $K$ of $k$-means instances, assuming that $K$ is learnable with 
a zero sample error, non-inventive learning algorithm. 
In the analysis, zero sample error and non-inventive properties play a key 
role in the crucial step of bounding the cost of incorrectly labeled points 
in terms of that of correctly labeled nearby points.

We now present two learning theorems for the Euclidean setting and the 
finite metric space setting. 
Assuming no point lies on cluster boundaries, 
the theorems state that the labeling function $f_\calO$ of the optimal 
clustering is learnable with a zero sample error, non-inventive 
learning algorithm in both settings. 
We say that a $k$-means instance $(P, Q, \dist)$ has \emph{no boundary points} 
if in any optimal clustering $\calO$ with clusters $O_1, \dots, O_k$ 
and respective centers $o_1, \dots, o_k$, the closest center to any given point 
$p \in P$ is unique (i.e., if $p \in O_i$, $\dist(p, o_i) < \dist(p, o_j)$ 
for all $j \neq i$).

\begin{theorem}[Learning k-Means in Euclidean Space]
\label{thm:learning-alg-euc}
Let $\dist(p_1, p_2) = \|p_1 - p_2 \|$ be the Euclidean distance function. Let 
$K = \{(P, \R^r, \dist) : P \subset \R, |P| < \infty, (P,\R^r,\dist) 
\text{ has no boundary points} \}$
be the family of $k$-means instances that are in $r$-dimensional Euclidean 
space and that have no boundary points.
The family $K$ is learnable with sample-complexity\footnote{$\tilde{O}$ 
hides $\poly(\log \log k, \log \log r)$ factors.}
$m(\R^r, \epsilon, \delta) = \tilde{O}((k^2r \log(k^2r) 
\allowbreak (\log ( {k^3r}/{\epsilon}))+ 
\log ({1}/{\delta}) )/\epsilon ).$
\end{theorem}

\begin{theorem}[Learning k-Means in Finite Metric Spaces]
\label{thm:learning-alg-fms}
Let $K = \{ (P, Q, \dist) : (P \cup Q, \dist) $ is a finite metric space, and 
$(P,Q,\dist) \text{ has no boundary points} \}$
be the family of finite metric space $k$-means instances that have no 
boundary points.
The family $K$ is learnable with sample-complexity\footnote{$\tilde{O}$ 
hides $\poly(\log \log k, \log \log |Q|)$ factors.}
$m( Q, \epsilon, \delta) = \tilde{O}( (k^2 (\log k) (\log |Q|) 
( \log k + \log {1}/{\epsilon} ) + 
\log ( {1}/{\delta} )) /\epsilon ).$
\end{theorem}

We prove Theorem~\ref{thm:learning-alg-euc} and Theorem~\ref{thm:learning-alg-fms} 
in Appendix~\ref{app:pac}, where we also introduce 
the necessary PAC learning concepts and tools.


\section{A Simple Algorithm for \texorpdfstring{$(1+\epsilon)$}{(1 + epsilon)} Cost and \texorpdfstring{$(1-\epsilon)$}{(1 - epsilon)} Accuracy}
\label{sec:simplealg}

Let $K$ be a family of $k$-means instances that is learnable with sample 
complexity $m$ using a zero sample error, non-inventive learning algorithm 
$\calA_L$. 
Let $\calA_\alpha$ be a constant-factor approximation algorithm 
(in terms of cost) for $k$-means, and let 
$\calA_1$ be a polynomial-time algorithm for the $1$-means problem (i.e., given 
$(P, Q, \dist) \in K$, $\calA_1$ finds $\argmin_{q \in Q} \cost(P, q)$ 
in polynomial time).
We present a simple semi-supervised learning algorithm that, given a 
$k$-means instance $(P, Q, \dist)$ of class $K$ and oracle access to the 
labeling function $f_\calO$ of a fixed optimal clustering $\calO$ of 
$(P, Q, d)$, outputs a clustering $\hat{\calO}$ that, with probability at 
least $(1 - \delta)$, satisfies $\cost(\hat{\calO}) \leq (1 + \epsilon)OPT$ 
and $\error(\hat{\calO}, \calO) \leq \epsilon |P|$.
Our algorithm uses $\calA_\alpha$, $\calA_1$, and $\calA_L$ as subroutines and 
makes $O((k\log |P|) \cdot m(Q, \epsilon^4, \delta/(k \log |P|)))$ 
oracle queries.
We show that our algorithm can be easily modified for 
$(1 + \epsilon)$-approximate and $(1 - \epsilon)$-accurate $k$-median and 
other similar distance-based clustering problems.
Towards the end of this section, we discuss several applications of this result, 
namely, for Euclidean and finite metric space $k$-means and $k$-median problems.

Let us start by applying
the learning algorithm $\calA_L$ to learn all the cluster labels. 
If we get perfect accuracy, the cost will be optimal. 
A natural question to ask in this case is: what happens to the cost if the 
learning output has $\epsilon$ error? 
In general, even a single misclassified point can incur an arbitrarily large 
additional cost. 
To better understand this, consider the following: 
Let $O_i, O_j \subseteq P$ be two distinct optimal clusters in the target 
clustering, and let $o_i$, $o_j$ be their respective cluster centers. 
Let $p \in O_i$ be a point that is incorrectly classified and assigned label 
$j \neq i$ by $\calA_L$. 
Also assume that the number of misclassified points is small enough 
so that the centers of the clusters output by the learning algorithm are 
close to those of the optimal clustering. 
Thus, in the optimal clustering, $p$ incurs a cost of $\distsq(p, o_i)$, 
whereas according to the learning outcome, $p$ incurs a cost that is close 
to $\distsq(p, o_j)$. 
In the worst case, $\dist(p, o_j)$ can be arbitrarily larger than 
$\dist(p, o_i)$.

Now suppose that, within distance $\rho$ from $p$, there exists some point 
$q \in O_j$. 
In this case, we can bound the cost incurred due to the erroneous label of 
$p$ using the true cost of $p$ in the target clustering. 
To be more specific, using the triangle inequality, we get the following 
bound for any metric space: 
$\dist(p, o_j) \leq \dist(p, q) + \dist(q, o_j) \leq \rho + \dist(q, o_j)$. 
Furthermore, due to the optimality, $\dist(q, o_j) \leq \dist(q, o_i) 
\leq \dist(q, p) + \dist(p, o_i) \leq \rho + \dist(p, o_i)$. 
Hence, it follows that $\dist(p, o_j) \leq 2 \rho + \dist(p, o_i)$. 
To utilize this observation in an algorithmic setting, we need to make sure 
that, for every point that is misclassified  into cluster $j$, there exists 
a correctly classified nearby point $q$ that belongs to the optimal cluster 
$O_j$. 
Luckily, this is ensured by the combination of zero sample error and 
non-inventive properties of $\calA_L$. 
If a point is misclassified into cluster $j$, the non-inventive property 
says that $\calA_L$ must have seen a sample point $q$ from cluster $j$. 
The zero sample error property ensures that $q$ is labeled correctly by 
$\calA_L$.
To make sure that such correctly labeled points are sufficiently close to 
their incorrectly labeled counterparts, we run $\calA_L$ separately on certain 
suitably bounded partitions of $P$. 

The formal description of our algorithm is given in Algorithm~\ref{alg:simple-alg}. 
The outline is as follows:
First, we run $\calA_\alpha$ on $(P, Q, \dist)$ and obtain an intermediate 
clustering $\calC = \{C_1, \dots, C_k\}$.
For each $C_i$, we run $\calA_1$ to find a suitable center $c_i$. 
Next, we partition each intermediate cluster $C_i$ into an inner ball and 
$O(\log |P|)$ outer rings centered around $c_i$.
We run the learning algorithm $\calA_L$ separately on each of these partitions. 
We choose the inner and outer radii of the rings so that, in each partition, 
the points that are incorrectly classified by the learning algorithm only incur 
a small additional cost compared to that of the correctly classified points.
The final output is a clustering $\hat{\calO}$ that is consistent with 
the learning outputs on each of the partitions. 
For each cluster $\hat{O}_i$, we associate the output of running $\calA_1$ on 
$(\hat{O}_i, Q, \dist)$ as its center.
Note that, due to the accuracy requirements, the cluster 
center to which a point is assigned in the output may not be the cluster center 
closest to that point in the output. 
It remains an interesting problem to find an accurate clustering
in which every point is always assigned to its nearest cluster center.

\begin{algorithm}[ht!]
  \SetKwInOut{Input}{Input}\SetKwInOut{Output}{Output}
  \Input{
	  $k$-Means instance $(P, Q, d)$, oracle access to $f_\calO$, 
	  constant-factor approximation algorithm $\calA_\alpha$ for $k$-means, 
	  $1$-means algorithm $\calA_1$, zero sample-error, non-inventive learning algorithm $\calA_L$ with 
	  sample complexity $m$, accuracy parameter $0 < \epsilon < 1$, 
	  and failure probability $0 < \delta < 1$.}
  \Output{
	  The clustering $\hat{\calO} = \{\hat{O}_1, \dots, \hat{O}_k\}$ 
	  defined by the labeling 
	  $f_{\hat{\calO}}: P \to [k]$ computed below. 
	  The respective cluster centers are 
	  $\hat{o}_i = \argmin_{q \in Q} \cost(\hat{O}_i, q)$, 
	  which can be found by running $\calA_1$ on $(\hat{O}_i, Q, \dist)$.}
  \BlankLine
  Let $n = |P|$, and let $\gamma = {\epsilon^2}/{(288 \alpha)}$. \\
  Run $\calA_\alpha$ and obtain an $\alpha$-approximate $k$-means clustering   
  $\calC = \{C _1, \dots, C _k\}$. 
  For each $i \in [k]$, run $\calA_1$ on $(C_i, Q, \dist)$ and find 
  centers $c_i = \argmin_{q \in Q} \cost(C_i, q)$. 
  \label{alg-step:const-approx}\\
  \For{$C_i \in \calC$}{
      Let $r_i = \sqrt{{\cost(C_i, c_i)}/{(\gamma |C_i|)}}$. 
      \label{alg-step:rings}\\
      Let $C_{i,0}$ be all points in $C_i$ that are at most $r_i$ away from 
      $c_i$. \\
      Let $C_{i, j}$ be the points in $C_i$ that are between $2^{j-1}r_i$ and 
	 		$2^jr_i$ away from $c_i$ for $j = 1, \dots, {(\log n)}/{2}$.  \\
      Let $m' = m( Q, \gamma^2, \delta/(k \log n))$.
       \label{alg-step:num-samples}\\
     \For{each non-empty $C_{i,j}$}{
        	Sample $m'$ points $x_1, \dots, x_{m'} \in C_{i,j}$ independently and 
				uniformly at random. \\
				Query the oracle on $x_1, \dots, x_{m'}$ and let 
				$S_{i,j} = \{(x_i, f_{\calO}(x_i)): i = 1, \dots, m'\}$. 
				\label{alg-step:sampling}\\
        Run $\calA_L$ on input $(P, Q, \dist)$ and $S_{i,j}$, 
        and obtain a labeling
        $h_{i,j}: C_{i,j} \to [k]$. \label{alg-step:pac-learning}
     }
  }
  Output the clustering $\hat{\calO}$ defined by the following labeling 
  function: \\
  \For{each $i$, $j$, $x \in C_{i,j}$ } { 
    Set $f_{\hat{\calO}}(x) = h_{i,j}(x)$.
  }
  \caption{A simple algorithm for $(1 + \epsilon)$-approximate 
  $(1 - \epsilon)$-accurate $k$-means.}
  \label{alg:simple-alg}
\end{algorithm}

We now analyze Algorithm~\ref{alg:simple-alg} and show that, with probability at least $(1 - \delta)$, 
it outputs a $(1 + \epsilon)$-approximate and $(1 - \epsilon)$-accurate $k$-means clustering.

Let $n = |P|$ be the total number of points in the $k$-means instance.
Assume that $0 < \gamma < \txtfrac{1}{2}$. 
For all $i, j, p$, let $H_{i,j,p} = \{ x \in C_{i,j}: h_{i,j}(x) = p\}$ be the 
set of points that are in $C_{i,j}$ and labeled $p$ by the output $h_{i,j}$ of 
the learning algorithm $\calA_L$. 
Call a point $x \in P$ \emph{bad} if $x \in H_{i,j,p}$ but $f_\calO(x) \neq p$; 
otherwise, call it \emph{good}. 
Denote the set of bad points by $B$ and let the complement $B^c$ 
of $B$ be the set of good points. 
For each $i$, let $o_i = \argmin_{q \in Q} \cost(O_i, q)$  
denote the center of cluster $i$ in $\calO$. 
For any point $x \in P$, let $o(x)$ denote the center of the optimal cluster 
for $x$ under the clustering $\calO$. 
Thus, $o(x) = o_p \Leftrightarrow f_\calO(x) = p$.

Notice that, for all $i$, all the points in $C_i$ belong to one of the 
$C_{i,j}$'s. 
In other words, no point in $C_i$ is more than $2^{(\log n)/2}r_i$ away from 
$c_i$, where $r_i =  \sqrt{\txtfrac{\cost(C_i, c_i)}{(\gamma |C_i|)}}$. 
To see this, suppose $x \in C_i$ is a point that is more than $2^{(\log n)/2} 
r_i$  away from $c_i$. 
Then, $\dist^2(x, c_i) \geq  2^{\log n} \cdot \txtfrac{\cost(C_i, c_i)}{(\gamma |C_i|)} 
> \cost(C_i, c_i)$ which is a contradiction.

\begin{lemma}
\label{lem:simple-correct}
With probability at least $(1 - \delta)$, all non-empty $C_{i,j}$'s 
satisfy $|C_{ij} \cap B| \leq \gamma^2|C_{ij}|$.
\end{lemma}
\begin{proof}
Recall that we run $\calA_L$ with $m_{K}(P, Q, \dist, \gamma^2, 
\txtfrac{\delta}{(k \log n)})$ samples. 
Thus, by definition, $\calA_L$, each run of $\calA_L$ succeeds with probability at 
least $\left( 1 - \txtfrac{\delta}{(k \log n)} \right)$. 
Since we only run $\calA_L$ at most $k \log n$ times, the claim follows from the 
union bound.
\end{proof}

We continue the rest of the analysis conditioned on  $|C_{i,j} \cap B| \leq 
\gamma^2 |C_{i,j}|$ for all  $C_{i,j}$. 
In proving the subsequent results, we use the following observations.

\begin{obsv}
\label{obsv:dist-lb}
No two points in $C_{i,j}$ are more than distance $R =  2 \cdot 2^{j}r_i$ 
apart. 
Note that, according to the definition of $C_{i,j}$ in the algorithm, $R$ is 
the outer diameter of the ring that bounds $C_{i,j}$.
\end{obsv}

\begin{obsv}
\label{obsv:cost-lb}
For $j \geq 1$, the inner radius of the ring that bounds $C_{i,j}$ is 
$2^{j-1}r_i$. 
Therefore, we have the following lower bound for the cost of $C_{i,j}$:
$\cost(C_{i,j}, c_i) \geq |C_{i,j}| (2^{j-1} r_i) ^ 2.$
\end{obsv}

\begin{lemma}\label{lem:dist_ineq} 
For all $i,j$ and $p$, if $x \in H_{i,j,p} \cap B$ then 
$\dist(x, o_p) \leq 4 \cdot 2^j r_i + \dist(x, o(x))$.
\end{lemma}

\begin{proof}
If $x \in H_{i,j,p} \cap B$, then  $x$ is in some optimal cluster denoted by 
$O_q$ for some $q \neq p$. 
Note that if $h_{i,j}$ (i.e., the output of the algorithm $\calA_L$) gives label 
$p$ to some point $x$, then the non-inventive property of $\calA_L$ ensures that 
it has seen at least one point $y \in C_{i,j}$ that is in $O_p$, and the 
zero sample-error property ensures that $y$ is labeled correctly by $h_{i,j}$. 
Thus, $o_p = o(y)$ and $o_q = o(x)$.
Hence, $y$ is a \emph{good} point with label $p$, and consequently, 
we have
\begin{align*}
\dist(x, o_p) &\leq \dist(x, y) + \dist(y, o_p) \leq \dist(x, y) + 
\dist(y, o_q) \leq  2 \dist(x, y) + \dist(x, o(x))\\ &\leq 4 \cdot 2^j r_i + \dist(x, o(x)),
\end{align*}
where the last inequality follows from Observation~\ref{obsv:dist-lb}.
\end{proof}

\begin{lemma}[Squared Triangle Inequality]\label{lem:triangle}
  For any $a,b \ge 0$ and $0 < \epsilon < 1$, we have
  $$(a+b)^2 \le (1+\epsilon) a^2 + \left(1 + \frac{1}{\epsilon} \right) 
b^2 \le (1+\epsilon) a^2 + \left(\frac{2}{\epsilon} \right) b^2\,.$$
\end{lemma}
\begin{proof}
The first inequality follows from the AM-GM inequality because 
$2 ab \leq \epsilon a^2 + \txtfrac{b^2}{\epsilon}$. 
The second inequality holds because $\epsilon < 1$ implies 
$1 < \txtfrac{1}{\epsilon}$.
\end{proof}

Now, let us analyze the cost of the labeling output by Algorithm~\ref{alg:simple-alg}:

\begin{align*}
\cost(P, \hat{\calO}) &= \sum_{p \in [k]} \cost(\hat{\calO}_p, \hat{o}_p) 
\leq \sum_{p \in [k]} \cost(\hat{\calO}_p, o_p)  = \sum_{i,j,p} \cost(H_{i,j,p}, o_p). \\
\intertext{Splitting the cost contributions of good and bad points, we get}
\cost(P, \hat{\calO})  &\leq \sum_{x \in B^c} \distsq(x, o(x)) + \sum_{i,j, p} 
\cost(H_{i,j,p} \cap B, o_p) \\
&=  \sum_{x \in B^c} \distsq(x, o(x)) + \sum_{i,j, p} 
\sum_{x \in H_{i,j,p} \cap B} 
\distsq(x, o_p).\\
\intertext{Applying Lemma~\ref{lem:dist_ineq} together with Lemma~\ref{lem:triangle}, for any $\epsilon \in (0, 1)$, we have}
\cost(P, \hat{\calO}) &\leq  \sum_{x \in B^c} \distsq(x, o(x)) + \sum_{i,j, p} 
\sum_{x \in H_{i,j,p} \cap 
B} \left( \left(1 + \frac{\epsilon}{3} \right)  \distsq(x, o(x)) + 
\frac{2 \cdot 3}{\epsilon} (4 \cdot 2^{j}r_i)^2 \right) \\
&\leq  \left(1 + \frac{\epsilon}{3} \right) \sum_{x \in P} \distsq(x, o(x)) + 
\sum_{i,j, p} \sum_{x \in H_{i,j,p} \cap B}  \frac{96}{\epsilon} (2^{j}r_i)^2 \\
&= \left(1 + \frac{\epsilon}{3} \right)  OPT + \sum_{i,j} \sum_{x \in C_{i,j} \cap B}  
\frac{96}{\epsilon} (2^{j}r_i)^2 .\\
\intertext{From Lemma~\ref{lem:simple-correct}, we have
  $|C_{i,j} \cap B| \leq \gamma^2 |C_{i,j}|$, and it follows that}
\cost(P, \hat{\calO}) &\leq  \left(1 + \frac{\epsilon}{3} \right)  OPT + 
\sum_{i,j} \gamma^2 | 
C_{i,j}|  \frac{96}{\epsilon} 
(2^{j}r_i)^2\\
&=\left(1 + \frac{\epsilon}{3} \right)  OPT + \sum_{i} \gamma^2 | C_{i,0}|  
\frac{96}{\epsilon} r_i^2 + \sum_{i,j: j \geq 1} \gamma ^ 2 | C_{i,j}|  
\frac{96}{\epsilon} (2^{j}r_i)^2. \numberthis \label{eqn:cost}
\end{align*}
Consider the last two terms of Equation~\eqref{eqn:cost} individually. 
For the first summation, we have
\begin{align*}
\sum_{i} \gamma^2 | C_{i,0}|  \frac{96}{\epsilon} r_i^2 &= \sum_{i} \gamma^2 | 
C_{i,0}|  \frac{96}{\epsilon} \frac{\cost(C_i, c_i)}{\gamma |C_i|} 
\leq \sum_{i} \frac{96 \gamma}{\epsilon} \cost(C_i, c_i)
\leq \frac{96 \gamma \alpha}{\epsilon}  OPT.
\end{align*}
In the last inequality, we used the fact that $\calC$ gives  an 
$\alpha$-approximation for the optimal cost.
For the second summation of Equation~\eqref{eqn:cost}, Observation~\ref{obsv:cost-lb} gives
\begin{align*}
\sum_{i,j: j \geq 1} \gamma ^ 2 | C_{ij}|  \frac{96}{\epsilon} (2^{j}r_i)^2 &= 
\sum_{i,j: j \geq 1} \frac{4 \cdot 96 \gamma^2 }{\epsilon} |C_{ij}|(2^{j-1} 
r_i)^2 
\leq \sum_{i,j: j \geq 1} \frac{384 \gamma^2 }{\epsilon} \cost(C_{ij}, c_i) \\
&\leq \sum_{i} \frac{384 \gamma^2 }{\epsilon} \cost(C_{i}, c_i) \leq \frac{384 \gamma^2 \alpha}{\epsilon} OPT.
\end{align*}
Here, we have again used the approximation guarantee of $\calC$ in the 
final inequality.

Choosing $\gamma =  \txtfrac{ \epsilon ^2}{(288 \alpha)} $ makes sure that both 
$ \txtfrac{96 \gamma \alpha}{\epsilon} \leq \txtfrac{\epsilon}{3}$ and $\txtfrac{384 \gamma^2 \alpha}{\epsilon} \leq  
\txtfrac{\epsilon}{3}$, and consequently, we get a  final cost of at most 
$(1 + \epsilon) OPT$. 
Recall that we established this bound conditioned on $|C_{i,j} \cap B| 
\leq \gamma^2 |C_{i,j}|$ for all $i,j$. 
In Lemma~\ref{lem:simple-correct}, we saw that all the $O(k \log n)$ runs of 
the learning algorithm $\calA_L$ succeed with probability at least $(1 - \delta)$.
Hence, the condition  $|C_{i,j} \cap B| \leq \gamma^2 |C_{i,j}|$ is true for 
all $i,j$ with the same probability. 
Summing the inequality over all $i,j$ yields $|B| \leq \gamma^2 |P| 
\leq \epsilon |P|$. 
Consequently, the output of Algorithm~\ref{alg:simple-alg}, with probability
at least $(1 - \delta)$ over the choice of samples in Step~\ref{alg-step:sampling}, outputs a $(1+\epsilon)$-approximate and 
$(1 - \epsilon)$-accurate $k$-means clustering.

In Algorithm~\ref{alg:simple-alg}, instead of an exact algorithm
$\calA_1$ for the $1$-means problem, we can also use a PTAS. 
Using a PTAS to approximate $1$-means up to a $(1 + \epsilon)$ factor will 
only cost an additional $(1 + \epsilon)$ factor in our cost analysis.
As a result, we get the same approximation and accuracy guarantees if we 
replace $\epsilon$ with ${\epsilon}/{3}$.

Algorithm~\ref{alg:simple-alg} makes $O((k \log n ) \cdot m(Q, \epsilon^4, 
\allowbreak \delta /(k \log n) ))$ queries
 to the oracle $f_\calO$ in total. Recall that simulating an oracle 
query to $f_\calO$ takes $O(k)$ same-cluster queries. Therefore, the total 
number of same-cluster queries
is $O((k^2 \log n ) \cdot m(Q, \epsilon^4, \allowbreak \delta /(k \log n) ))$.

Our definition of a learning algorithm in Section~\ref{sec:pac-for-kmeans} has 
nothing to do with whether the input is a $k$-means instance or a 
$k$-median instance, which is similar to $k$-means except that the cost of 
a cluster $C$ with respect to a center $q$ is defined as 
$\cost(C, q) := \sum_{p \in C} \dist(p,q)$.
In fact, it applies to any similar clustering scenario where the cost is defined
in terms of the $\ell$'th power ($\ell > 0$) of distances instead of 
squared distances.
The analysis of Algorithm~\ref{alg:simple-alg} can be adapted to any 
fixed $\ell$ once we have a suitable triangle inequality 
analogous to Lemma~\ref{lem:triangle}. 
For example, when $\ell \leq 1$, we can simply use the trivial inequality 
$(a + b)^\ell \leq a^\ell + b^\ell$. 
Thus, for such clustering problems, Algorithm~\ref{alg:simple-alg}, with a 
slight modification on choice of radii in Step~\ref{alg-step:rings} and a 
little adjustment to the parameter $\gamma$, will 
give the same guarantees. Hence, we have the following 
theorem which is the formal version of Theorem~\ref{thm:simplealg-informal}. 
The proof follows from the analysis of Algorithm~\ref{alg:simple-alg}.

\begin{theorem}
\label{thm:simplealg}
Let $K$ be a family of $k$-means ($k$-median) instances. Suppose that $K$ is 
learnable with sample complexity 
$m(Q, \epsilon, \delta)$ using a zero sample error, non-inventive 
learning algorithm $\calA_L$. 
Let $\calA_\alpha$ be a constant-factor approximation algorithm, and 
let $\calA_1$ be a 
PTAS for the $1$-means ($1$-median) problem. 
There exists a polynomial-time algorithm that, given an instance 
$(P, Q, \dist) \in K$, oracle access to same-cluster queries for 
some fixed optimal clustering $\calO$, and parameters 
$(\epsilon, \delta) \in (0, 1)^2$, 
outputs a clustering 
that, with probability at least $(1 - \delta)$, is $(1 - \epsilon)$-accurate 
with respect to $\calO$, and simultaneously has a cost of at most 
$(1 + \epsilon) OPT$. 
The algorithm uses $\calA_L$, $\calA_\alpha$, and $\calA_1$ as subroutines.
The number of same-cluster queries made by the algorithm is
\begin{enumerate}
\item $O( (k^2 \log |P|)\, \cdot m( Q, \epsilon^4, 
{\delta}/{(k \log |P|)} ))$ for the $k$-means setting and
\item $O( (k^2 \log |P|)\, \cdot m( Q, \epsilon^2, 
{\delta}/{(k \log |P|)}))$ for the $k$-median setting.
\end{enumerate} 
\end{theorem}

For $k$-means and $k$-median instances in Euclidean space and 
those in finite metric spaces, there exist several constant-factor 
approximation algorithms 
(for example, Ahmadian et al.~\cite{ANSW17} and Kanungo et al.~\cite{KMNPSW02}). 
Solving the $1$-means problem in Euclidean space is straightforward: 
The solution to $\argmin_{q \in \R^r} \cost(P, q)$ is simply 
$q = {(\sum_{p \in P} p )}/{|P|}$. 
For the $k$-median problem in Euclidean space,
the problem of $1$-median does not have an exact algorithm but 
several PTASes exist (for example, Cohen et al. \cite{CLMPS16}).
In a finite metric space, to solve $\argmin_{q \in Q} \cost(P, q)$, we can 
simply try all possible $q \in Q$ in polynomial time, and this holds for 
the $k$-median setting as well. 
Thus, for Euclidean and finite metric space $k$-means and 
$k$-median instances that have no boundary points, 
Theorem~\ref{thm:simplealg}, 
together with Theorem~\ref{thm:learning-alg-euc} and Theorem~\ref{thm:learning-alg-fms}, gives efficient algorithms 
for $(1 +\epsilon)$-approximate, 
$(1 - \epsilon)$-accurate semi-supervised clustering.


\section{Removing the Dependency on Problem Size in the Query Complexity for Euclidean \texorpdfstring{$\mathbf{k}$}{k}-Means}
\label{sec:alg2}

For the family of Euclidean $k$-means instances, the query complexity of 
Algorithm~\ref{alg:simple-alg} suffers from a $\tilde{O}(\log n)$ dependency 
(where $n$ is the number of points in the input $k$-means instance, and
$\tilde{O}$ hides $\poly (\log \log n)$ factors)
due to the repeated use of the learning algorithm $\calA_L$.
Specifically, we run $\calA_L$ with a failure probability of 
${\delta}/{(k \log n)}$, $O(\log n)$ times per cluster.  
Note that the sample complexity of $\calA_L$ itself, in the case of Euclidean 
$k$-means instances, does not have this dependency.

In this section, we show that we can avoid this dependency on $n$ using a 
slightly more involved algorithm at the cost of 
increasing the query complexity by an extra $\poly(k)$ factor. 
Nevertheless, this algorithm has superior performance when the size of the 
input instance (i.e., the number of points) is very large 
(when $\log n = \Omega(k^{10})$ for example).

Recall that, for a set $C \subset \R^r$, $\cost(C, y)$ is minimized when $y$ 
is the centroid of $C$, denoted by $\mu(C) = (\sum_{x \in C} x)/|C|$. 
Define the fractional size of an optimal cluster $O_i$ as the fraction of 
points that belong to $O_i$, i.e., the ratio ${|O_i|}/{n}$.
Suppose we only want to get a good approximation for the cost, and that 
we know that all the clusters in the target solution have sufficiently large 
fractional sizes. 
In this case, naive uniform sampling will likely pick a large number of 
samples from each of the clusters. 
This observation, together with Lemma~\ref{lem:approx-centroid}, allows us 
to approximate the centroid and the cost of each cluster to any given accuracy.

\begin{lemma}[Lemma 1 and Lemma 2 of Inaba et al.~\cite{IKI94}]
\label{lem:approx-centroid}
  Let $(\epsilon, \delta) \in (0, 1)^2$, let $m \geq {1}/{(\epsilon \delta)}$ 
  be a positive integer, and let  $S = \{ p_1, \ldots, p_m \}$ be a multiset 
  of $m$ i.i.d. samples from the uniform distribution over some finite set 
  $C \subset \R^r$.
  With probability at least 
  $(1 - \delta)$, $\distsq(\mu(S), \mu(C)) \leq \epsilon \cdot 
  {\cost(C, \mu (C))}/{|C|}$ and $\cost(C, \mu(S)) \leq (1 + \epsilon) 
  \cost(C, \mu(C))$.
\end{lemma}

However, the above approach fails when some clusters in the optimal target 
solution contribute  significantly to the cost, but have small fractional 
sizes (that is because uniform sampling is not guaranteed to pick sufficient 
numbers of samples from the small clusters).
Ailon et al.~\cite{ABJK18} circumvented this issue with an algorithm 
that  iteratively approximates the centers of the clusters using a 
distance-based probability distribution ($D^2$-sampling). 
We will refer to their algorithm as $\calA^\ast$.

Note that when it comes to accuracy, we can totally disregard clusters 
with small fractional sizes; we only have to correctly label a sufficiently 
large fraction of the points in large clusters. 
With this intuition, we present the outline of our algorithm. 

Let $(P, \R^r, \dist)$ be a $k$-means instance in Euclidean space that has 
no boundary points. 
For simplicity, we refer to the instance $(P, \R^r, \dist)$ by just 
$P$ where possible, as for Euclidean $k$-means, the other two parameters 
are fixed.
We start with a naive uniform sampling step that gives a good approximation 
for the centers of large clusters. 
Starting with these centers, we run a slightly modified version of 
algorithm $\calA^\ast$ to approximate the centers of the remaining small 
clusters. 
Thus, at this point, we have a clustering with a good cost and we know which 
clusters are large.
We now run the learning algorithm $\calA_L$ on input $P$ 
and obtain a labeling of the points. 
For each point, we assign its final label based on \begin{enumerate}
\item the label assigned to it by the learning algorithm $\calA_L$, and
\item its proximity to large cluster centers.
\end{enumerate}
In particular, if the output of $\calA_L$ decides that a point $p$ should be 
in some large cluster $i$, and if $p$ is sufficiently close to the approximate 
center for cluster $i$, we label it according to the learning output; 
otherwise, we label it according to its nearest approximate center.
We show that this approach retains a cost that is close to the cost of the 
clustering output by $\calA^\ast$. 
The accuracy guarantee comes from the facts that a large fraction of the 
points are sufficiently close to the centers of large clusters, and that 
$\calA_L$ labels most of them correctly with a good probability. 

We now review the key properties of algorithm $\calA^\ast$ 
(the algorithm of Ailon et al.~\cite{ABJK18}).
Let $0 < \epsilon < 1$. We say a $k$-means instance $P$ is 
$(k, \epsilon)$-irreducible if  no $(k-1)$-means 
clustering gives an $(1 + \epsilon)$-approximation for the $k$-means problem, 
i.e., if $OPT^{k}$ denotes the optimal $k$-means cost of $P$, then $P$ is $(k, 
\epsilon)$-irreducible if  $OPT^{k-1} > (1 + \epsilon) OPT^k$.
Suppose that $P$ is $(k, \epsilon)$-irreducible. Let 
$\calO = \{O_1, \dots, O_k\}$ be the target optimal clustering, and let 
$o_1, \dots, o_k$ be the respective centers.
Let $C_i = \{c_1, \dots, c_i \}$ denote a set of $i$ centers and let 
$Z(i)$ denote the following statement: 
There exists a set of $i$ distinct indices $j_1, \dots, j_i$ 
such that, for all $r \in [i]$, $ \cost(O_{j_r}, c_r) \leq (1 + 
{\epsilon}/{16} )  \cost(O_{j_r}, o_{j_r})$. 
To put it differently, $Z(i)$ says that $C_i$ is a set of good candidate 
centers for $i$-many distinct clusters in the target optimal solution. 
Assuming $P$ is $(k, \epsilon)$-irreducible, the algorithm  $\calA^\ast$  
yields a method to incrementally construct sets 
$C_1, \dots, C_k$ (i.e., $C_{i + 1} = C_i \cup \{c_{i + 1} \}$) such that,  
conditioned on $Z(i)$ being true, $Z(i + 1)$ is true with probability 
at least $(1 - {1}/{k} )$. 
Now suppose that $P$ is $(k, {\epsilon}/{(4k)})$-irreducible. 
Then $\calA^\ast$ gives a $(1 + {\epsilon}/{(4 \cdot 16 k)})$-approximation 
for $k$-means with probability at least $(1 - 1/k)^k \geq 1/4$.
Otherwise, $\calA^\ast$ gives a 
$(1 + {\epsilon}/{(4 \cdot 16 k)})$-approximation 
for the $i$-means problem for some $i < k$, where $i$ is the largest 
integer such that $P$ is $(i, {\epsilon}/{4k})$-irreducible. 
In the latter case, it will give a $(1 + {\epsilon}/{(4 \cdot 16 k)}) 
( 1 + \epsilon/(4k))^{k - i}$-approximation with probability at least $1/4$.
In either case, the output of $\calA^\ast$ is a $(1 + \epsilon)$-approximation.

In our algorithm, we first find the centers of large clusters using 
uniform sampling, and
then run $\calA^\ast$ to find the remaining centers. 
This allows us to know which clusters are large,
which is a crucial information needed for the final labeling.
Suppose that in the target optimal solution we have $k_0 \leq k$ clusters 
whose fractional sizes are at least $\epsilon/k$. 
Note that $k_0$ is at least $1$ due to the Pigeonhole Principle, since at 
least one cluster should have a fractional size of at least $1/k > 
\epsilon/k$. 
By Lemma~\ref{lem:approx-centroid}, using uniform sampling, we can 
approximate the centroid of each of these large clusters with  a good accuracy. 
Hence, we can have a set $C_{k_0}$ of $k_0$ centers such that $Z(k_0)$ is true 
with probability $(1 - \delta)$. 
Afterwards, we use $\calA^\ast$ to incrementally construct 
$C_{k_0 + 1}, \dots, C_{k}$. 
Conditioned on $Z(k_0)$ being true, the output $C_k$ will be a
$(1 + \epsilon)$-approximation with probability 
$(1 - 1/k)^{k - k_0} \geq (1 - 1/k)^{k} \geq 1/4$ 
for $k \geq 2$.
However, by independently running this incremental construction 
$O(\log(1/\delta))$ times and choosing the set of centers with the 
minimum total cost, we can boost this probability to $(1 - \delta)$. 
This observation gives the following generalization of Theorem 10 of 
Ailon et al.~\cite{ABJK18}.

\begin{theorem}
\label{thm:cost-approx}
Consider a Euclidean $k$-means instance $(P, \R^r, \dist)$.  
Let $O_1, \dots, O_k$ be a fixed optimal clustering with respective centers 
$o_1, \dots, o_k$.
Let $k_0 \leq k$ and let $C_{k_0} = \{c_1, \dots, c_{k_0}\}$ be a set of points 
such that, with probability at least $p_0$,  
$\cost(O_i, c_i) \leq (1 + {\epsilon}/{(64 k)}) \cost(O_i, o_i)$ 
for all $i \in [k_0]$. There exists an algorithm 
$\calA_{cost}$ that, given $P$, $C_{k_0}$, and 
parameters $(\epsilon, \delta) \in (0, 1)^2$ as 
input, outputs a set of centers $C_{k} = C_{k_0} \cup \{c_{k_0 + 1}, 
\dots, c_k \}$ such that 
$\sum_{i \in [k]} \cost(O_i, c_i) \leq (1 + \epsilon) \sum_{i \in [k]} 
\cost(O_i, o_i)$ with probability at least $p_0(1 - \delta).$ 
Moreover, $\calA_{cost}$ uses 
$O(({k^{9}}/{\epsilon^4}) \log({1}/{\delta}))$ 
same-cluster queries and runs in time $O(nr({k^{9}}/{\epsilon^4})
\allowbreak \log({1}/{\delta}))$.
\end{theorem}

Theorem~\ref{thm:cost-approx} implies a method to get a good approximation for 
the cost that also reveals which clusters are large. 
Specifically, we first perform uniform sampling over the whole set $P$ and 
approximate the centers of the large clusters. 
If we get a sufficient number of samples, the approximate centers will satisfy 
the precondition of Theorem~\ref{thm:cost-approx} with a good probability. 
Thus, using algorithm $\calA_{cost}$ from Theorem~\ref{thm:cost-approx}, 
we get the 
desired approximation for the cost. 
What remains now is to use PAC learning and to appropriately label the points 
according to the learning outcome. 

\begin{algorithm}
  \SetKwInOut{Input}{Input}\SetKwInOut{Output}{Output}
  \Input{Point set $P \subset \R^r$, the oracle access to $f_\calO$, 
  parameter $k$, accuracy parameter $\epsilon$, failure probability $\delta$, 
  and algorithms $\calA_{cost}$ and $\calA_L$.}
  \Output{The clustering $\hat{\calO} = \{\hat{O}_1, \dots, \hat{O}_k\}$ 
                defined by the labeling 
	  $f_{\hat{\calO}}: P \to [k]$ computed below. For each $i \in [k]$, 
	  the respective	cluster center $\hat{o}_i$ is the centroid of $\hat{O}_i$.}
  \BlankLine
  Draw $Q_1(k,\epsilon,\delta)$ samples from $P$ independently and uniformly 
  at random,
  and query $f_\calO$ to get their true cluster labels in $\calO$. 
  Denote the set of sampled points by $S$, and for all $i \in [k]$, denote 
  the set of sampled points that belong to class $i$ by $S_i$. 
  \label{step:unif-sample}\\
  Let $k'$ be the number of distinct cluster labels with more than 
  $ ({\epsilon}/{(2 k)})Q_1(k,\epsilon,\delta) $ samples.
  Let $C_{k'} := \{ \mu(S_i) : |S_i| >  ({\epsilon}/{(2k)}) 
   Q_1(k,\epsilon,\delta) \}$. 
  Without loss of generality, assume that the class labels for centers  
  in $C_{k'}$ are $1, \ldots, k'$. \label{step:large-centers}\\
  Run the algorithm $\calA_{cost}$, starting from $C_{k'}$ as the partial
  set of centers.  
  This takes $Q_2(k,\epsilon,\delta)$ more queries. 
  Let $C_k = \{c_1, \dots, c_k\}$ be the output, and let $OPT^\ast$ be the 
  cost of the clustering obtained by assigning each point to its nearest $c_i$.
  \label{step:approx-label}\\
  Use the PAC learning algorithm $\calA_L$ on $Q_3(k, r, 
  {\epsilon^4}/{k}, 
  \delta)$ uniform i.i.d. samples from $P$ to learn a classifier for the $k$ 
  classes that is  $(1-{\epsilon^4}/{k} )$-accurate with 
  probability at least $(1 - \delta)$. 
  Let $H_1, \ldots, H_{k}$ be the sets of points that are  labeled 
  $1, \dots, k$ respectively by the classifier. 
  \label{step:pac} \\
  Output the clustering $\hat{\calO}$ defined by the following labeling 
  function: for each $i \in [k']$ and $p \in H_i$ such that 
  $\distsq(p, c_i) \le {k OPT^\ast}/{(n \epsilon^3)}$, set 
  $f_{\hat{\calO}}(p) = i$.  
  For any other point $p$, set $f_{\hat{\calO}}(p) = i$ if the nearest 
  cluster center to $p$ in $C_k$ is $c_i$.
  \label{step:recolor}
  \caption{Algorithm whose  query complexity is independent of $n$}
  \label{alg:indep-n}
\end{algorithm}

We present the pseudo-code of our algorithm in Algorithm~\ref{alg:indep-n}, 
where we use the algorithm $\calA_{cost}$ from Theorem~\ref{thm:cost-approx} 
and the learning algorithm $\calA_L$ guaranteed by Theorem 
\ref{thm:learning-alg-euc}.  
In Algorithm~\ref{alg:indep-n}, 
$Q_1(k, \epsilon,\delta)={256 k^3}/{(\epsilon^2 \delta)}$ is
the number of samples needed to ensure that we pick a sufficient number of 
samples from each of the clusters with fractional sizes of at least 
${\epsilon}/{k}$,
$Q_2(k,\epsilon,\delta)$ is the sample complexity of the algorithm 
$\calA_{cost}$, and $Q_3(k,r,\epsilon,\delta)= m(\R^r,\epsilon,\delta)$ 
is the sample complexity of the learning algorithm $\calA_L$ for an  
error $\epsilon$ and
a failure probability $\delta$. 
As with Algorithm~\ref{alg:simple-alg}, the center that a point is assigned 
to in the final output may not be the closest center to that point.

We now prove that, with probability at least $(1 - \delta)$, the output of Algorithm~\ref{alg:indep-n} is $(1 + \epsilon)$-approximate and $(1 - \epsilon)$-accurate.

Assume $0 < \epsilon < \txtfrac{1}{4}$. 
For an optimal cluster $O_i$ with center $o_i$, denote by 
$\avgdistsq(O_i):=\txtfrac{\cost(O_i, o_i)}{|O_i|}$ the average squared distance from 
the points in $O_i$ to their center $o_i$.  
Let $c_1, \dots, c_{k'}$ be the points in $C_{k'}$ where $c_i$ is the approximate 
centroid for the cluster with label $i$ found in Step~\ref{step:large-centers} of 
Algorithm~\ref{alg:indep-n}.
  
First, we show that Step~\ref{step:large-centers} of Algorithm~\ref{alg:indep-n} approximates all the large cluster centers accurately enough 
so that the precondition for applying Theorem~\ref{thm:cost-approx} is 
satisfied (recall that the precondition is to have a set of $k' < k$ 
approximate centers for $k'$ distinct optimal clusters in the target 
solution).
Lemma~\ref{lem:precond} ensures this, and also makes sure that the clusters 
that are too small are not picked as large clusters. 
This last fact is useful in the proof of Lemma~\ref{lem:acc-and-cost}.

\begin{lemma} \label{lem:precond}
With probability at least $1-\delta$, the following statements are true:
\begin{enumerate}
	\item For all $i \in [k']$, $\distsq(o_i, c_i) \leq 
	(\txtfrac{\epsilon}{(64k)}) \avgdistsq(O_i)$.
	\item For all $i \in [k']$, $\cost(O_i, c_i) \leq \left(1 + 
	\txtfrac{\epsilon}{(64k)} \right) \cost(O_i, o_i)$.
	\item Let $L = \{i \in [k]: |O_i| \geq (\txtfrac{\epsilon}{k})n \}$ 
	be the set of the labels of the optimal clusters with a fractional 
    size of at least $\txtfrac{\epsilon}{k}$. 
    Then, $L \subseteq [k']$.
	\item Let $T = \{i \in [k]: |O_i| \leq (\txtfrac{\epsilon^2}{k})n \}$ be 
	the set of the labels of the optimal clusters with fractional size 
	of at most $\txtfrac{\epsilon^2}{k}$. Then $T \cap [k'] = \emptyset$.
\end{enumerate}
\end{lemma}

\begin{proof}
Notice that for each $i\in[k']$, we have at least 
$(\txtfrac{\epsilon}{(2k)}) Q_1(k, \epsilon, \delta) = \txtfrac{128k^2}{(\epsilon \delta)}$ 
samples in Step~\ref{step:large-centers}. 
Thus, using  Lemma~\ref{lem:approx-centroid} on each cluster with error 
parameter $\txtfrac{\epsilon}{(64k)}$ and failure probability $\txtfrac{\delta}{(2k)}$, 
and applying the union bound, the first two statements hold with probability 
at least $(1 - \txtfrac{\delta}{2} )$. 

As for the final two statements, note the following.  
Let $q =  Q_1(k, \epsilon, \delta)$, and let the random variables 
$X_{i, j}: j = 1, \dots, q$ be such that $X_{i, j} = 1$ if the $j$-th sample 
is from $O_i$. Otherwise, $X_{i, j} = 0$. 
Let $X_i = \sum_{j \in [q]} X_{i,j}$ be total number of 
samples picked from $O_i$ and $p_i = \Pr[X_{i,j} = 1]$. 
Since we pick identical samples, this probability does not depend on $j$.

If $i \in L$, then $p_i \geq \txtfrac{\epsilon}{k}$ and 
$\E[X_i] = q \cdot p_i \geq \txtfrac{q \cdot \epsilon}{k}$.  
Applying a standard Chernoff bound, we get
\begin{align*}
\Pr \left[ X_i < \frac{1}{2} \cdot \frac{\epsilon}{k} q\right]  
\leq \Pr \left[ X_i < \frac{1}{2} \E[X_i] \right] 
\leq   \exp \left(-\frac{ q \cdot \epsilon}{2^2 \cdot 3 \cdot k} \right) 
\leq  \exp \left(-\frac{32 k^2}{3 \epsilon \delta} \right) \leq \frac{\delta}{2k}.
\end{align*}
For $i \in T$, observe that $p_i \leq \txtfrac{\epsilon^2}{k}$ and 
$\E[X_i] = q \cdot p_i \leq \txtfrac{q \cdot \epsilon^2}{k}$. 
Since $\epsilon < \txtfrac{1}{4}$, we further have that 
$\txtfrac{\epsilon}{(2k)} > 2\txtfrac{\epsilon^2}{k}$. 
Applying another standard Chernoff bound, we now get
\begin{align*}
\Pr \left[ X_i > \frac{1}{2} \cdot \frac{\epsilon}{k} q \right] 
& \leq \Pr \left[ X_i > 2 \cdot \frac{\epsilon^2}{k} q \right] 
\leq   \Pr \left[ X_i > 2  \E[X_i] \right] \\  
&\leq \exp \left(-\frac{ q \cdot \epsilon^2}{2^2 \cdot 3 \cdot k} \right) 
\leq  \exp \left(-\frac{32 k^2}{3  \delta} \right) \leq \frac{\delta}{2k}.
\end{align*}

Applying the union bound, we get that the final two statements are also true 
with probability at least $\left(1 - \txtfrac{\delta}{2}\right)$.  
Thus, all four statements are true with probability at least $(1 - \delta)$.
\end{proof}

Recall that $OPT^\ast$ is the cost of the clustering defined by the centers 
$C_k$ after Step~\ref{step:approx-label} of Algorithm~\ref{alg:indep-n}. 
Lemma~\ref{lem:precond}, together with Theorem~\ref{thm:cost-approx}, implies 
that, with probability at least $(1 - 2 \delta)$, 
$$ OPT^\ast \leq \sum_{i \in [k]} \cost(O_i, c_i) \leq (1 + \epsilon)\sum_{i \in [k]} 
\cost(O_i, o_i) \leq (1 + \epsilon) OPT.$$
Also note that the PAC learning of Step~\ref{step:pac}  has at most 
$\txtfrac{\epsilon^4}{k}$ error  with probability at least $(1 - \delta)$. 

The following lemma thus provides the final ingredient of our proof. 
Conditioned on Steps~\ref{step:unif-sample}, \ref{step:large-centers}, 
\ref{step:approx-label}, and~\ref{step:pac} of Algorithm~\ref{alg:indep-n} 
being successful,  Lemma~\ref{lem:acc-and-cost}  shows that the clustering 
defined in Step~\ref{step:recolor} achieves the desired accuracy of the 
algorithm with only a marginal increase in the cost.

\begin{lemma}
\label{lem:acc-and-cost}
Consider an outcome of Algorithm~\ref{alg:indep-n} where all four statements 
of Lemma~\ref{lem:precond} are true. 
Additionally, suppose that 
$ \sum_{i \in [k]} \cost(O_i, c_i) \leq (1 + \epsilon) OPT$, 
and that the PAC learning of Step~\ref{step:pac} of Algorithm~\ref{alg:indep-n} 
has at most $\txtfrac{\epsilon^4}{k}$ error.
Then, the output clustering $\hat{\calO}$ 
is $(1 - 6 \epsilon)$-accurate with respect to the target optimal clustering $\calO$
and $(1 + 3 \epsilon)$-approximates the $k$-means cost.
\end{lemma}

\begin{proof}

It is easy to see that the points that are incorrectly labeled by 
$f_{\hat{\calO}}$ are fully contained inside the union of the following 
three disjoint sets:
\begin{enumerate}
\item Points that belong to the target optimal clusters whose fractional 
sizes are at most $\txtfrac{\epsilon}{k}$. 
\item Points that are incorrectly labeled by the clustering algorithm in Step~\ref{step:pac} but are not in the set defined above.
\item Points that do not satisfy the distance criterion in Step~\ref{step:recolor} but are correctly classified and are assigned  a label from 
$\{1, \dots, k'\}$ by the PAC learning output in Step~\ref{step:pac}.
\end{enumerate}

At most a $k (\txtfrac{\epsilon}{k}) = \epsilon$ fraction of the points are in the 
first set.
By the accuracy of the PAC learning step, the second set has at most 
$\txtfrac{\epsilon^4}{k}$ fraction of the points. 
As for the third set, we consider the following:
By  the fourth statement of Lemma~\ref{lem:precond}, 
for $i = 1, \dots, k'$, $O_i$ has a fractional size of at least 
$\txtfrac{\epsilon^2}{k}$; therefore,  
$\txtfrac{\cost(O_i, o_i)}{(n (\epsilon^2 / k))} \geq  \avgdistsq(O_i)$.  
Recall that $H_i$ is the set of points that are assigned label $i$ 
in the PAC learning step. Let $i \in \{1, \dots, k' \}$, and let $p \in H_i$ 
be a point that is correctly labeled in the PAC learning output and that 
satisfies $\distsq(p, c_i) > \txtfrac{k OPT^\ast}{(n \epsilon^3)} $. 
For such a point $p$, we can lower bound the distance to its approximate 
center at follows:\begin{align*}
\distsq(p, c_i) &> \frac{k OPT^\ast}{n \epsilon^3} \geq \frac{k OPT}{n \epsilon^3} \geq
 \frac{k \cost(O_i, o_i)}{n \epsilon^3} \geq  \frac{\avgdistsq(O_i)}{\epsilon}.
\end{align*}
Also, since $2 \distsq(p, o_i) + 2 \distsq(o_i, c_i) \geq \distsq(p, c_i)$, 
and  $\distsq(o_i, c_i) \leq (\txtfrac{\epsilon}{(64k)}) \avgdistsq(O_i)$ by 
Lemma~\ref{lem:precond} (1), we get
$$ \distsq(p, o_i) > \frac{1}{2} \left( \frac{\avgdistsq(O_i)}{\epsilon} - 
2 \distsq(o_i, c_i) \right) >  \frac{1}{2} \left( \frac{\avgdistsq(O_i)}{\epsilon} - 
2 \frac{\epsilon \avgdistsq(O_i)}{64k} \right)  > \frac{\avgdistsq(O_i)}{4 \epsilon}.$$ 
Hence, there can be at most $4 \epsilon |O_i|$ such points. 
Summing over all $O_i$, we conclude that at most $4 \epsilon$ fraction 
of the points are in the final set.
Hence, the resulting clustering is at least $(1 - 6\epsilon)$ accurate.

It remains to show that the cost increase in Step~\ref{step:recolor} is small. 
Let $ \cost' = \sum_{i \in [k]} \cost(O_i, c_i)$ be the cost of assigning each 
optimal cluster $O_i$ to its approximate center $c_i$. 
We know $cost'$ is close to the optimal cost by Theorem~\ref{thm:cost-approx}. 
We  now show that the label assignment in Step~\ref{step:recolor} does not increase 
this cost by too much.
Observe that for the points that are not labeled according to the PAC learning output, 
the contribution to $\cost'$ can only decrease by assigning it to the nearest $c_i$. 
For the points that are labeled according to the PAC output, if the learning 
algorithm has assigned them the correct labels, then there is no change in 
their contribution to $\cost'$.

Thus, it is sufficient to bound the cost of those points that are 
incorrectly labeled by the PAC learning and are labeled according to the 
learning output in the final assignment. 
There can be at most $(\txtfrac{\epsilon^4}{k}) n$ such points due 
to the learning accuracy, and assigning such a point $p$ to $c_i$ can 
incur at most an extra $\distsq(p, c_i)$ cost. 
However, due to the distance constraint of Step~\ref{step:recolor} of 
Algorithm~\ref{alg:indep-n}, $\distsq(p, c_i) \leq \txtfrac{k OPT^\ast}{(n \epsilon^3)}$. 
Therefore, the total cost increase by these points is at most 
$$\frac{\epsilon^4  n}{k} \frac{k OPT^\ast}{n \epsilon^3} \leq \epsilon 
OPT^\ast \leq  \epsilon (1 + \epsilon) OPT \leq 2 \epsilon OPT.$$
Consequently, the cost of the output clustering of the algorithm is at most 
$(1 + 3 \epsilon) OPT$.
\end{proof}

Combining Lemma~\ref{lem:precond}, Theorem~\ref{thm:cost-approx}, and Lemma~\ref{lem:acc-and-cost}, we have the proof of Theorem~\ref{thm:indep-n}.
We need each of the Steps~\ref{step:unif-sample}, \ref{step:approx-label}, 
and~\ref{step:pac} to succeed with probability at least $(1 - \txtfrac{\delta}{3})$ 
so that we get a final failure probability of $\delta$, which
can be easily achieved by replacing the probability parameter $\delta$ of Algorithm~\ref{alg:indep-n} with $\delta/3$. 
Furthermore, according to the statement of Lemma~\ref{lem:acc-and-cost}, we also 
need to replace the accuracy parameter $\epsilon$ by $\txtfrac{\epsilon}{6}$. 
As for the claim on the query complexity, we recall that we only need $O(k)$ 
same-cluster queries per single $f_\calO$ query, and Algorithm~\ref{alg:indep-n} 
makes a total number of $Q_1(k,\epsilon,\delta) + Q_2(k,\epsilon,\delta) 
+ Q_3(k, r, \txtfrac{\epsilon^4}{k}, \delta)$ queries to $f_\calO$.


\section{Lower Bounds for \texorpdfstring{$\mathbf{(1 - \epsilon)}$}{(1 - epsilon)}-Accurate \texorpdfstring{$\mathbf{k}$}{k}-Means }

\label{sec:lbs}

The results we presented in Section~\ref{sec:simplealg} and Section~\ref{sec:alg2}
uses PAC learning to achieve $(1-\epsilon)$-accuracy. 
Considering the class of Euclidean $k$-means instances, the query complexities 
of both our results have a linear dependency on the dimension of the Euclidean space. 
This is in contrast to the query complexities of the algorithms in 
Ashtiani et al.~\cite{AKB16}, which had strong assumptions on the input, and Ailon et al.~\cite{ABJK18}, 
which was only aiming to approximate the optimal cost.
In this section, we argue that the linear dependency on dimension is 
necessary for accuracy for Euclidean $k$-means instances. 
This result, as shown in the latter half of this section, also implies that 
the dependency on $\log|Q|$ is necessary for the $k$-means instances in finite 
metric spaces, where $Q$ is the set of candidate centers. 

Let $r>0$ be an integer, and $e_1, \ldots, e_r$ be the standard basis in $\R^r$.
Consider the $2$-means instance $(P, \R^r, \dist)$ in $r$-dimensional Euclidean space,
where $P =\left\{-e_1, e_1, \ldots, -e_r, e_r \right\}$. There are $2r$ points in $P$.
In Lemma~\ref{lem:opt-clust}, we show that any optimal solution for this instance 
contains exactly one of $\pm e_i$ in each optimal cluster for all $i=1, \ldots, r$. 
Hence, there are $2^r$ different optimal solutions.
This means that, without querying at least one of $\pm e_i$, it is 
information-theoretically impossible to know which cluster $e_i$ and
$-e_i$ each belong to, and thus to achieve constant classification error,
the query complexity must be linear in $r$. This proves the first claim of
Theorem~\ref{thm:lbs}.

\begin{lemma}
\label{lem:opt-clust}
Consider the $k$-means instances $(P, \R^r, \dist)$ defined as above. 
Let $\calO = \{O_1, O_2\}$ be any fixed optimal clustering on $(P, \R^r, \dist)$. 
For each $i \in [r]$, either $e_i \in O_1$ and $-e_i \in O_2$ or $e_i \in O_2$ and $-e_i \in O_1$.
\end{lemma}

\begin{proof}
Consider a bi-partition $A \dot{\cup} B = \{ \pm e_1, \ldots, \pm e_r\}$. 
We first observe the following: Let $i \in [r]$ be such that $-e_i \in A$ and $e_i \in B$.
Define $A':=A \setminus \{-e_i\} \cup \{e_i\}$ and
$B':=B \setminus \{e_i\} \cup \{-e_i\}$.
Then the cost of the solution $(A,B)$ is the same as the cost of $(A',B')$.
This is because $e_i$ and $-e_i$ are at the same distance to all points in 
$P \setminus \{\pm e_i\}$.

Therefore, without loss of generality, we only need to consider 
bi-partitions where the set $A$ can be described as
$ A = \{e_{r_0+1}, \ldots, e_{r_0+r_1}\} \cup 
  \{\pm e_{r_0+r_1+1}, \ldots, \pm e_{r_0+r_1+r_2}\} \,,
$
for some $r_0$, $r_1$ and $r_2$ where $r_0+r_1+r_2=r$.
Consequently, we have that
$
  B = \{\pm e_{1}, \ldots, \pm e_{r_0}\} \cup
  \{-e_{r_0+1}, \ldots, -e_{r_0+r_1}\} \,.
$
To put differently, $A$ contains only the positive points on axes $r_0 + 1, \dots, r_0 + r_1$
whereas $B$ contains only the negative points on the same axes. 
For axes $r_0 + r_1 + 1, \dots, r_0 + r_1 + r_2$, $A$ contains both positive and 
negative points on those axes. 
Similarly, for axes $1, \dots, r_0$, $B$ contains both points on those axes. 
The center of $A$ is
\[
(
\underbrace{0, \ldots, 0}_{r_0}, 
\underbrace{\frac{1}{r_1+2r_2}, \ldots, \frac{1}{r_1+2r_2}}_{r_1},
\underbrace{0, \ldots, 0}_{r_2}
) \,,
\]
and the center of $B$ is
\[
(
\underbrace{0, \ldots, 0}_{r_0}, 
\underbrace{-\frac{1}{r_1+2r_0}, \ldots, -\frac{1}{r_1+2r_0}}_{r_1},
\underbrace{0, \ldots, 0}_{r_2}
) \,.
\]
The cost of set $A$ is
\begin{align*}
&~ r_1 \cdot \left( \left(1-\frac{1}{r_1+2r_2}\right)^2 + (r_1-1) \frac{1}{(r_1+2r_2)^2}\right) 
+ 2r_2 \cdot \left( 1 + r_1 \frac{1}{(r_1+2r_2)^2}\right) \\
=&~ \frac{1}{(r_1+2r_2)^2} \left(
  r_1 (r_1+2r_2-1)^2 + r_1 (r_1-1) + 2r_2 (r_1+2r_2)^2 + 2r_1r_2
  \right) \\
=&~ \frac{1}{(r_1+2r_2)^2} \left(
  r_1 (r_1+2r_2)^2 - 2 r_1(r_1+2r_2) + r_1 + r_1 (r_1-1) + 2r_2 (r_1+2r_2)^2 + 2r_1r_2
  \right) \\
=&~ \frac{1}{(r_1+2r_2)^2} \left( (r_1+2r_2)^3 - r_1 (r_1 + 2 r_2) \right) \\
=&~ (r_1+2r_2) - \frac{r_1}{r_1+2r_2}\,.
\end{align*}
Similarly, the cost of set $B$ is
\begin{align*}
&~ r_1 \cdot \left( \left(1-\frac{1}{r_1+2r_0}\right)^2 + (r_1-1) \frac{1}{(r_1+2r_0)^2}\right) 
+ 2r_0 \cdot \left( 1 + r_1 \frac{1}{(r_1+2r_0)^2}\right) \\
=&~ (r_1+2r_0) - \frac{r_1}{r_1+2r_0}\,.
\end{align*}
The total cost of this clustering is therefore
\[
  2r - \left(\frac{r_1}{r_1+2r_2} + \frac{r_1}{r_1+2r_0}\right)\,.
\]
It remains to show that the maximum of 
\begin{equation*}
  \left(\frac{r_1}{r_1+2r_2} + \frac{r_1}{r_1+2r_0}\right)
  = \left(\frac{r_1}{r_1+2r_2} + \frac{r_1}{2r-2r_2-r_1}\right)
\end{equation*}
is achieved by setting $r_1=r$ and $r_0=r_2=0$. For a fixed $r_1$, the last expression gives
a convex function in variable $r_2$, and therefore either $r_2=0$ or $r_2=r-r_1$ maximizes
the expression. In either case, we have that the maximum of
this expression for a fixed $r_1$ is
\[
1 + \frac{r_1}{2r-r_1} = \frac{2r}{2r-r_1}\,,
\]
which in turn achieves maximum when $r_1=r$.
\end{proof}
 
Now consider the hard example for Euclidean setting with $r = (\log n / 2)$. 
Thus, we have $\log n$ points in $P$. 
Start with an empty set $Q$, and for each subset $P'$ of $P$, add 
$\mu(P') = \txtfrac{(\sum{p \in P'} p)}{|P'|}$ to $Q$.
This forms a finite metric space $k$-means instance $(P, Q, \dist)$ that
has no boundary points.
From the claim on dimension dependency for the Euclidean setting
discussed above, it follows that $\Omega(r) = \Omega(\log n) = \Omega (\log |Q|)$ samples 
are necessary to achieve $(1 - \epsilon)$-accuracy. 
This proves the second claim of Theorem~\ref{thm:lbs}.


\section*{Acknowledgments}
This research was supported by ERC Starting grant 335288-OptApprox.


\bibliography{references}


\appendix

\section{PAC Learning of \texorpdfstring{$\mathbf{k}$}{k}-Means Clustering Instances}
\label{app:pac}

In this section, we introduce some fundamental concepts and tools from 
PAC learning theory, and prove Theorem~\ref{thm:learning-alg-euc} and 
Theorem~\ref{thm:learning-alg-fms}. 
We start by introducing some necessary notations.

Let $\calX$ be an arbitrary domain set, $\calY$ be a label set, and
$\calH \subseteq \{ h:\calX \to \calY \}$ be a hypothesis class.
Let $\calD$ be some arbitrary distribution over $\calX \times \calY$.
The error of a hypothesis $h \in \calH$ with respect to $\calD$ is defined as
\begin{equation}
\loss{\calD}{h} := \Pr_{(x,y) \sim \calD} \left[ h(x) \ne y \right]\,. 
\label{eq:pac-error}
\end{equation}

We say that $\calH$ satisfies \emph{the realizability assumption} if there 
exists some $h^\ast \in \calH$ such that $L_{\calD}(h^\ast) = 0$. 
A learning algorithm $\calA$ for hypothesis class $\calH$ receives as input 
a sequence $S:=\left( (x_1,y_1), \dots, (x_m,y_m) \right)$ of $m$ pairs of 
domain points and their labels, where each domain point is sampled 
independently from $\calD$. 
The learning algorithm should output a predictor $\calA(S) \in \calH$. 
The goal of the algorithm is to minimize the generalization error 
$\loss{\calD}{\calA(S)}$ with respect to the unknown $\calD$. 
We define the empirical error (risk) on $S$ for a hypothesis $h$ as
\begin{equation}
\loss{S}{h} := \frac{\left| \left\{ i \in [m]: h(x_i) \ne y_i \right\}\right|}{m}\,.
\end{equation}
We call the algorithm $\calA$ an \emph{empirical risk minimization (ERM)} 
algorithm if $$ \calA(S) \in \argmin_{h \in \calH} \loss{S}{h}\,.$$

Let $\calH$ be a hypothesis class for some discrete domain $\calX$. 
If the realizability  assumption holds for $\calH$, and if $\calA$ is an 
ERM algorithm for $\calH$, then $L_S(\calA(S)) = 0$. 

Clearly, $\calA(S)$ depends on $S$. 
Hence, the generalization error, 
$L_D(A(S))$, is a random variable whose randomness depends on $S$. 
We are interested in establishing an upper bound on the number of samples that 
$\calA$ needs in order to guarantee that the generalization error is small with 
a good probability.

\begin{definition}[Sample Complexity of a Learning Algorithm]
Let  $\calA$ be a learning algorithm for a hypothesis class $\calH$. 
We define the sample complexity $m_{\calA,\calH}(\epsilon, \delta)$ of $\calA$
as the minimum natural number such that the following holds for all distributions 
$\calD$ over $\calX \times \calY$: If $S$ is sequence of $m \geq 
m_{\calA, \calH}(\epsilon, \delta)$ i.i.d. samples from $\calD$, with 
probability at least $1 - \delta$, $\loss{\calD}{\calA(S)} \leq \epsilon$.
\end{definition}

We now introduce some terminology from learning theory.
We first define the notions of \emph{shattering} and \emph{VC-dimension} for 
a \emph{binary} hypothesis class $\calH \subseteq \{h: \calX \to \{1, 2\}\}$. 

\begin{definition}[Shattering]
\label{def:shat}
We say a binary hypothesis class $\calH$ shatters a finite set  
$${C=\left\{ c_1,\dots,c_m \right\} \subseteq \calX},$$ if 
$\left\{ (h(c_1),\dots,h(c_m)): h \in \calH \right\} = \left\{ 1, 2 \right\}^{|C|}\,$.
\end{definition}

\begin{definition}[Vapnik–Chervonenkis (VC) Dimension]
The VC-dimension of a binary hypothesis class $\calH$, denoted $\vcdim(\calH)$,
is the maximum size of a set $C \subset \calX$ that can be shattered by 
$\calH$.
If $\calH$ shatters sets of arbitrarily large size, we say that $\calH$ has 
infinite VC-dimension.
\end{definition}

By definition, to shatter a set $C$, a hypothesis class $\calH$ must have at 
least $2^{|C|}$ distinct elements. 
Hence, we have the following lemma.

\begin{lemma}
\label{lem:vc-dim-finite}
Let $\calH$ be a finite hypothesis class. 
Then, $\vcdim(\calH) = O(\log |\calH|)$.
\end{lemma}

For a $k$-ary hypothesis class $\calH \subseteq \{h: \calX \to [k]\}$, 
we define the analogous notions of \emph{multiclass shattering} and 
\emph{Natarajan dimension}. 
As we will see below, the notion of multiclass shattering is a generalization 
of the notion of shattering for the binary hypothesis classes. 
Consequently, the Natarajan dimension is a generalization of the VC-dimension.

\begin{definition}[Multiclass Shattering]
\label{def:multshat}
We say that a set $C \subseteq \calX$ is shattered by a $k$-ary hypothesis 
class $\calH$ if there exist two functions $f_0, f_1: C \rightarrow [k]$ 
such that,
\begin{itemize}
\item For every $x\in C$, $f_0(x) \neq f_1(x)$.
\item For every $B \subset C$, there exists a function $h \in \calH$ such that,
$$\forall x \in B, h(x)=f_0(x) \mbox{ and } \forall x \in C \backslash B, h(x)=f_1(x).$$
\end{itemize}
\end{definition}

\begin{definition}[Natarajan Dimension]
The Natarajan dimension of a $k$-ary hypothesis class $\calH$, denoted by 
$\Ndim(\calH)$, is the maximal size of a set $C \subseteq \calX$ that can 
be shattered by $\calH$.
\end{definition}

Suppose $k = 2$. If a binary hypothesis class $\calH_{\bin}: \calX \to \{1, 2\}$ 
shatters a set $C \subseteq \calX$ according to Definition~\ref{def:shat},  
then for all subsets $B$ of $C$, there exists a hypothesis $h \in \calH_{\bin}$ 
that maps all points in $B$ to $1$ and all points in $C \backslash B$ to $2$. 
Thus, setting $f_0(x) = 1$ and $f_1(x) = 2$ to be constant functions, it 
follows that $\calH_{\bin}$ shatters $C$ according to Definition~\ref{def:multshat} as well.

For the  converse,  suppose that $\calH_{\bin}$ is a hypothesis class that 
shatters a set $C \subseteq \calX $ according to Definition~\ref{def:multshat}. 
Definition~\ref{def:multshat} essentially says that, any mapping 
$C \to \{1 , 2 \}^{|C|}$ must be achievable with one of the hypotheses in 
$\calH_{\bin}$ which is equivalent to Definition~\ref{def:shat}.  
To see this, notice that $f_0(x) \neq f_1(x)$, and that we only have two labels.

In the remainder of this appendix, we prove Theorem~\ref{thm:learning-alg-euc} 
and Theorem~\ref{thm:learning-alg-fms}. 
We proceed with defining a new hypothesis class for binary classification which 
is associated with non-homogeneous halfspaces in $\R^r$. 
Formally, given a non-homogeneous hyperplane $\ell \subset \R^r$, let 
$h_\ell$ be the binary classifier that assigns labels $\pm1$ to points based on 
which side of the hyperplane the point lies on. 
We define the binary hypothesis class $\calH^{\euc}_{\bb}$ as
$$\calH^{\euc}_{\bb} = \{h_\ell : \ell \mbox{ is a (non-homogeneous) hyperplane in } \R^r\}.$$
 
\begin{lemma}[Theorem 9.3 of Shalev-Shwartz and Ben-David~\cite{SB14}] 
\label{prop:vcdim-of-halfspaces} 
The VC-dimension of $\calH^{\euc}_{\bb}$ is $r + 1$.
\end{lemma} 

Let $(P, \R^r, \dist)$ be a $k$-means instance in Euclidean space. 
Let $\calO$ be a fixed optimal clustering of $P$, and let $f_\calO$ be the 
labeling function of $\calO$. 
Assume that $(P, \R^r, \dist)$ has no boundary points (i.e., in an optimal 
clustering, for any given point $p \in P$, the closest optimal center is 
unique). 
Let $A, B \subset P$ two non-empty disjoint subsets of points of two 
different labels under $f_\calO$. 
Due to the assumption of having no boundary points, there exists a hypothesis 
$h \in \calH^{\euc}_{\bb}$ that perfectly separates points in $A$ from those in $B$.
That is to say, it assigns $-1$ to all points in $A$ and $+1$ to all points in 
$B$. Such a hypothesis $h$ can be learned in polynomial-time using standard
SVM (Support Vector Machine) training algorithms like the one proposed by Joachims~\cite{Joa99}.

Let $S = \{ (x_1, y_1), \dots , (x_m, y_m) \}$ be a multiset of sample-label 
pairs from $P \times Q$. 
Fix a pair of distinct cluster labels $a, b \in [k']$ and let $S_{a, b}$ be 
the subset of the samples-label pairs in $S$ whose labels are either $a$ or $b$. 
Since each pair of original clusters can be separated with a hyperplane in 
$\R^r$, 
there exists a hypothesis $g_{a,b} \in \calH^{\euc}_{b}$ that perfectly labels the 
samples in $S_{a, b}$; that is to say that for all $(x, y) \in S_{a, b}$, 
$g_{a,b}(x) = 1$ if $y = a$ and $g_{a,b}(x) = -1$ if $y = b$. 
Suppose for each pair $a, b \in [k]$, $a \ne b$, we find a hypothesis $g_{a,b}$.
Then, for all samples $(x, a) \in S$, $\sum_{b \neq a} g_{a, b}(x) = k - 1$. 
Moreover, for any sample $(x, a) \in S$, and for any $a' \neq a$, 
$\sum_{b \neq a'} g_{a', b}(x) < k - 1$. 
This is because $g_{a', a}(x) = -1$. 
Consequently, the function $g$ defined as 
$$ g(x) = \argmax_{a \in [k]} \sum_{b \in [k], b \neq a} g_{a, b}(x)$$ 
returns the correct label $y$ for all $(x, y) \in S$.
Hence, the aforementioned procedure gives an ERM algorithm for the hypothesis 
class $\calH^\ast$ which is formally defined follows:
for 
$\bar{g} = (g_{a, b})_{[a, b \in [k], a \neq b]} \in \left(\calH^{\euc}_{\bb}\right)^{k(k-1)}$, 
a $(k(k-1))$-tuple of binary hypotheses, let 
$h_{\bar{g}}(x) = \argmax_{a \in [k]} \sum_{b \in [k], b \neq a} g_{a, b}(x).$
Define
$$ \calH^\ast = \left\{ h_{\bar{g}} : \bar{g} \in \left(\calH^{\euc}_{\bb}\right)^{k(k-1)}\right\}. $$
 
The pseudo-code of the algorithm is presented in Algorithm~\ref{alg:mult-from-bin}. 
This algorithm is an adaptation of the All-Pairs algorithm described in 
Chapter 17.1 of Shalev-Shwartz and Ben-David~\cite{SB14}. 
An observant reader may notice that learning two hypotheses, namely $g_{a,b}$ 
and $g_{b,a}$, per each distinct $a, b$ is redundant as one can set 
$g_{a, b}(x) = - g_{b, a}(x)$. 
However, for the ease of presentation, we stick with the idea that $g_{a,b}$ 
and $g_{b,a}$ are picked independently from each other.

\begin{algorithm}
  \SetKwInOut{Input}{Input}\SetKwInOut{Output}{Output}
  \Input{A set of samples $S = (x_1, y_1), \dots, (x_m, y_m)$ and an 
  ERM algorithm $\calB$ for learning $\calH^{\euc}_{\bb}$.}
  \Output{A hypothesis $g \in \calH^\ast$.}
  \BlankLine
  \For{$a, b \in [k]$ such that $a \neq b$}{
     Let $S_{a,b} = [\, ]$ be an empty list.\\
     \For{$t = 1, \dots, m$}{
        If $y_t= a$ add $(x_t, 1)$ to $S_{a,b}$. \\
        If $y_t = b$ add $(x_t, -1)$ to $S_{a,b}$. 
     }
     Let $g_{a,b} = \calB(S_{a,b})$.
  }
  Let 
  $g(x) = \argmax _{a \in [k] } \left( \sum_{b \in [k], b \neq a} g_{a, b}(x) \right)$.
  \label{alg:mult-from-bin}
  \caption{An ERM algorithm for $k$-category classification.}
\end{algorithm}

The corollary to the following lemma bounds the Natarajan dimension of 
$\calH^\ast$.

\begin{lemma}[Lemma 29.6 of Shalev-Shwartz and Ben-David~\cite{SB14}]
\label{lem:nat-dim-of-mult-class}
Consider a multiclass predictor derived in the following way. 
Train $l$ binary classifiers from a binary hypothesis class $\calH_{\bin}$ 
and let $v:\{-1,1\}^l \rightarrow [k]$ be a mapping from the $l$ predictor 
results to a class label. 
Let $\calH_{\bin}^v$ be the class of multiclass predictors obtained in this 
manner. 
If $\vcdim(\calH_{\bin}) = d$ then, $\Ndim(\calH_{\bin}^v) \leq 3 d l \log (ld)$.
\end{lemma}

\begin{corollary}
\label{cor:nat-dim-k2-hyp}
The Natarajan dimension of the hypothesis class $\calH^\ast$ is 
$O(rk^2 \log(rk))$ where $r$ is the dimension of the Euclidean space and 
$k$ is the number of clusters in the $k$-means instance.
\end{corollary}

The following lemma, combined with Corollary~\ref{cor:nat-dim-k2-hyp} 
yields that 
$$m_{\erm, \calH^\ast}(\epsilon, \delta) = O \left(\frac{k^2r \log (k^2 r) 
\left( \log \left(\frac{k^3 r}{\epsilon}\right) \right) + 
\log \left( \frac{1}{\delta} \right)} {\epsilon} \right).$$

\begin{lemma}[Theorem 3.7 of Daniely et al.~\cite{DSBS15}]
\label{lem:m-erm-ndim}
Let $\calH \subseteq \{h: \calX \to [k]\}$ be a hypothesis class, for which the 
realizability assumption holds with respect to some distribution $\calD$, and 
let $d = \Ndim(\calH)$. 
Let $$m_{\erm, \calH}(\epsilon, \delta) = \sup_{\calA \in \erm} 
m_{\calA, \calH}(\epsilon, \delta),$$ where the supremum is taken over 
all ERM algorithms for $\calH$. 
Then, 
$$m_{\erm, \calH}(\epsilon, \delta) = O \left(\frac{d \left( \log 
\left(\frac{k d}{\epsilon}\right) \right) + 
\log \left( \frac{1}{\delta} \right)} {\epsilon} \right)\,.$$
\end{lemma}

Thus, if $m \geq m_{\erm, \calH^\ast}(\epsilon, \delta)$ i.i.d. samples from 
$Q$ and their respective labels were input to Algorithm~\ref{alg:mult-from-bin}, 
the output hypothesis $h$, with probability at least $(1 - \delta)$, will have 
at most an $\epsilon$ error.

To prove Theorem~\ref{thm:learning-alg-euc}, what remains is to ensure that 
the learned hypothesis does not output class labels that the learning algorithm 
has not seen in the samples. 
For this, we propose a simple modification to Algorithm~\ref{alg:mult-from-bin}. 
That is, if $a'$ is a class label that is not present in the samples, for all 
class labels $a$ from which we have seen at least $1$ sample, we set 
$g_{a,a'}(x) = 1$ and  $g_{a', a}(x) = -1$ to be constant functions which may 
correspond to hyperplanes that are infinitely far away from the origin. 
Without loss of generality, assume that we have seen samples with labels 
$1, \dots, k'$, and we have not seen samples with labels $k' + 1, \dots, k$. 
Let $a \in k'$ be  a label that that we have seen and let 
$a' \in [k] \backslash [k']$ be a label that we have not seen.  
Therefore we have 
$$\sum_{b \neq a}g_{a,b}(x) = \sum_{b \in [k'], b \neq a} g_{a,b}(x) + 
\sum_{b \in [k] \backslash [k']}g_{a, b}(x) \geq  - (k' - 1) +  
(k - k')   = k - 2k' + 1,$$ and 
$$ \sum_{b \neq a'}g_{a',b}(x) = \sum_{b \in [k']} g_{a',b}(x) + 
\sum_{b \in [k] \backslash [k'], b \neq a'}g_{a', b}(x) \leq -k' +  
(k - k' - 1)   = k - 2k' - 1.$$ 
Consequently, the output $g(x)$ of Algorithm~\ref{alg:mult-from-bin} 
will never assign a label from $[k] \backslash [k']$ to any $x \in \R^d$.

To prove Theorem~\ref{thm:learning-alg-fms}, all we need to do is to 
replace the binary hypothesis class $\calH^{\euc}_{\bb}$ used in the preceding 
discussion with the binary hypothesis class $\calH^{\fms}_{\bb}$ which we introduce 
next. We show that $\calH^{\fms}_{\bb}$ satisfies the realizability assumption. 
We further show that there exists an $\erm$ algorithm for  $\calH^{\fms}_{\bb}$ and
the VC-dimension of $\calH^{\fms}_{\bb}$ is small. 

Let $(P, Q, \dist)$ be $k$-means instance in a finite metric space where  
$|Q| < \infty$, and define the binary hypothesis class $\calH^{\fms}_{\bb}$ 
as follows:
$$\calH^{\fms}_{\bb} = \{h_{q_1, q_2} : (q_1, q_2) \in Q \times Q  \}, $$ where 
$h_{q_1, q_2} : P \to \{-1, +1\}$ is a binary labeling function define as
$$ h_{q_1, q_2}(p) = 
  \begin{cases}
    -1, & \text{if } \dist(p, q_1) \geq \dist(p, q_2)  \\
  +1, & \text{otherwise. } 
  \end{cases}$$
Observe that this is a finite hypothesis class with at most $|Q|^2$ 
hypotheses. 
Furthermore, based on the assumption on no boundary points, for any two 
subsets $P$ of two different labels, there exists an $h \in \calH^{\fms}_{\bb}$ 
that perfectly separates them. 
Hence, pairwise binary classification of any two labels is realizable under 
$\calH_{\bb}$. 
To find such $h$ in polynomial-time, we simply iterate over all 
possible hypotheses and pick the one that gives zero error.
Lemma~\ref{lem:vc-dim-finite} yields that the VC-dimension of $\calH^{\fms}_{\bb}$ 
is $O(\log |Q|)$.

\end{document}